\documentclass[11pt, DIV=15]{scrarticle}

\usepackage[T1]{fontenc}

\usepackage[utf8]{luainputenc}
\usepackage{lmodern}
\usepackage{microtype}
\usepackage{amsmath,amssymb,amsthm,mathtools}
\usepackage{enumitem}
\usepackage{subcaption}

\usepackage{abstract}

\usepackage{thmtools, thm-restate}

\usepackage{xcolor}
\definecolor{darkblue}{rgb}{0.0, 0.0, 0.5}
\usepackage[colorlinks,linkcolor=darkblue,citecolor=darkblue,urlcolor=darkblue]{hyperref}
\usepackage[nameinlink,capitalize]{cleveref}

\usepackage{todonotes}

\usepackage[osf,sc]{mathpazo}

\newtheorem{theorem}{Theorem}[section]
\newtheorem{lemma}[theorem]{Lemma}

\newtheorem{fact}[theorem]{Fact}

\newtheorem{observation}[theorem]{Observation}
\newtheorem{corollary}[theorem]{Corollary}

\newtheorem{example}[theorem]{Example}

\theoremstyle{definition}
\newtheorem{definition}[theorem]{Definition}

\theoremstyle{remark}
\newtheorem{claim}[theorem]{Claim}

\newenvironment{claimproof}[1][Proof of Claim]{\begin{proof}[#1] }{ \end{proof}}

\DeclareMathOperator{\surj}{surj}
\DeclareMathOperator{\dwl}{dwl}

\newcommand{\den}[1]{\mathopen{[\![}#1\mathclose{]\!]}}

\usepackage{stmaryrd}

\newcommand{\col}[1]{\mathsf{#1}}
\newcommand{\SlidingTokens}{\textnormal{\textsc{SlidingTokens}}}

\usepackage{relsize,xspace}
\newcommand{\MWA}{\textsmaller{MWA}\xspace}
\newcommand{\MTA}{\textsmaller{MTA}\xspace}
\newcommand{\MWAs}{\textsmaller{MWA}s\xspace}
\newcommand{\MTAs}{\textsmaller{MTA}s\xspace}

\DeclareMathOperator{\poly}{poly}

\newcommand{\multiset}[1]{\{\!\!\{#1 \}\!\!\}}
\title{A Dense Weisfeiler--Leman Algorithm for Deciding Bounded-Cliquewidth \\ Homomorphism Indistinguishability}

\usepackage{authblk}

\author[1,2]{Radu Curticapean}
\author[3]{Daniel Neuen}
\author[1]{Amir Nikabadi}
\author[1]{Tim Seppelt}
\author[1]{Ben Young}
\affil[1]{IT University Copenhagen, Denmark}
\affil[2]{University of Regensburg, Germany}
\affil[3]{TU Dresden, Germany}

\usepackage[style=numeric-comp,maxbibnames=99,maxcitenames=99,backend=biber,bibencoding=utf8,sorting=none,giveninits=true]{biblatex}
\AtEveryBibitem{
	\clearfield{urldate}
	\clearfield{urlyear}
	\clearfield{urlmonth}
	\clearfield{publisher}
	\clearfield{issn}
	\clearfield{isbn}
	\clearfield{editor}
	\clearfield{editors}
	\clearlist{editors}
	\clearlist{editor}
	\clearfield{month}
	\clearfield{day}
	\clearlist{language}
	\clearlist{language}
	\clearlist{Language}
	\clearlist{langid}
}
\DeclareSourcemap{
	\maps[datatype=bibtex]{
		\map[overwrite]{
			\step[fieldsource=doi, final]
			\step[fieldset=url, null]
			\step[fieldset=eprint, null]
			\step[fieldset=isbn, null]
		}
	}
}
\DeclareDelimFormat{finalnamedelim}{\ifnumgreater{\value{liststop}}{2}{\finalandcomma}{}\addspace\&\space}

\DefineBibliographyExtras{UKenglish}{\def\finalandcomma{\addcomma}}

\setlist[enumerate, 1]{font=\upshape, noitemsep, nolistsep}
\setlist[enumerate, 2]{font=\upshape, noitemsep, nolistsep}
\setlist[itemize, 1]{noitemsep, nolistsep,font=\upshape}
\setlist[itemize, 2]{noitemsep, nolistsep,font=\upshape}

\usepackage{tikz}
\usetikzlibrary{decorations.markings}
\usetikzlibrary{calc}

\usepackage{csquotes}
\addtokomafont{disposition}{\normalfont\bfseries} 

\renewcommand\phi\varphi

\begin{document}
\maketitle
\begin{abstract}
    Two graphs $G$ and $H$ are \emph{homomorphism indistinguishable} over a graph class $\mathcal{F}$ if they admit the same number of homomorphisms from every graph in $\mathcal{F}$. A wide range of relaxations of graph isomorphism arise this way: isomorphism itself over the class of all graphs [Lov\'asz, Acta Math.\ Hung.\ 1967], equivalence under the $k$-dimensional Weisfeiler--Leman algorithm over the graphs of treewidth~$\leq k$ [Dvořák, J.\ Graph Theory 2010], and quantum isomorphism over planar graphs [Mančinska--Roberson, FOCS 2020]. Since the class $\mathcal{F}$ is typically infinite, it is not clear a priori whether homomorphism indistinguishability over $\mathcal{F}$ is decidable; for planar graphs it is undecidable. Every class for which decidability was previously known is sparse.

We give the first decidability results for dense graph classes: We introduce the \emph{dense Weisfeiler--Leman algorithm} that decides homomorphism indistinguishability over the class of graphs of cliquewidth~$\leq k$, the dense counterpart of treewidth. This relation was not previously known to be decidable. 
The algorithm colors $k$-tuples of vertex subsets rather than $k$-tuples of vertices.

Beyond the class of all graphs of cliquewidth~$\leq k$, 
we prove a general meta-theorem: homomorphism indistinguishability over every $\mathsf{CMSO}_1$-definable graph class of bounded cliquewidth is decidable, in randomized exponential time. For classes of bounded linear cliquewidth the bound improves to $\mathsf{PSPACE}$, and we show this is tight by exhibiting such a class for which the problem is $\mathsf{PSPACE}$-complete. These are the first general algorithms for homomorphism indistinguishability over dense graph classes.
\end{abstract}

\section{Introduction}

Deciding whether two graphs are isomorphic represents one of the most elusive problems in complexity theory.
Although a breakthrough result of \textcite{babai_graph_2016} placed graph isomorphism in quasipolynomial time, it is still not known to admit polynomial-time algorithms or to be in $\mathsf{coNP}$.
This lack of complexity-theoretic understanding and the fact that graph isomorphism is too brittle w.r.t.\ noise or perturbations arising in practice has led to a sustained interest in graph isomorphism relaxations, i.e., equivalence relations~$\approx$ between graphs that do not distinguish isomorphic graphs, see \cite{grohe_thoughts_2025}.

In the past, graph isomorphism relaxations have been studied in isolation, lacking a unified theory that explains their computational complexity.
Recently, \emph{homomorphism indistinguishability} has been proposed as such a theory, see \cite{seppelt_homomorphism_2024}.
Two graphs $G$ and $H$ are \emph{homomorphism indistinguishable} over a graph class~$\mathcal{F}$, in symbols $G \equiv_{\mathcal{F}} H$, if $G$ and $H$ admit the same number of homomorphisms from every graph $F \in \mathcal{F}$.
This notion subsumes a wide range of graph isomorphism relaxations:
\textcite{lovasz_operations_1967} showed that two graphs are isomorphic iff they are homomorphism indistinguishable over all graphs,
\textcite{dvorak_2010_homomorphisms} showed that two graphs are not distinguished by the $k$-dimensional Weisfeiler--Leman algorithm iff they are homomorphism indistinguishable over all graphs of treewidth~$\leq k$, and \textcite{mancinska_quantum_2020} showed that two graphs are quantum isomorphic iff they are homomorphism indistinguishable over all planar graphs. 
Moreover, a long list of graph isomorphism relaxations from algebraic graph theory, category theory, invariant theory, machine learning, and quantum information theory have been characterized as homomorphism indistinguishability relations over natural graph classes 
\cite{atserias_sherali-adams_2012, grohe_homomorphism_2025, roberson_lasserre_2024, grohe_pebble_2015, schindling_homomorphism_2025, montacute_pebble-relation_2024, dawar_lovasz-type_2021, abramsky_discrete_2022, grohe_counting_2020, fluck_seppelt_spitzer_2024, atserias_expressive_2021, rattan_weisfeiler-leman_2023, dell_grohe_rattan_2018, morris_weisfeiler_2019, xu_how_2019, zhang_beyond_2024, gai_homomorphism_2025, cerny_caterpillar_2025, cerny_homomorphism_2026, young_converse_2025, cai_vanishing_2025, kar_npa_2026, seppelt_quantum_2025}.\footnote{For further references, visit the \emph{Homomorphism Indistinguishability Zoo} at \href{https://tseppelt.github.io/homind-database}{tseppelt.github.io/homind-database}.}

The complexity of homomorphism indistinguishability relations can vary drastically. 
\textcite{boker_complexity_2019} constructed a bounded-treewidth graph class~$\mathcal{F}$ with polynomial-time membership test
for which~$\equiv_{\mathcal{F}}$ is undecidable.
Moreover, by the aforementioned characterizations, deciding $\equiv_{\mathcal{T}}$ for the class~$\mathcal{T}$ of all trees is in almost linear time, while $\equiv_{\mathcal{P}}$ for the class~$\mathcal{P}$ of all planar graphs is undecidable \cite{slofstra_set_2019,atserias_quantum_2019,mancinska_quantum_2020},
and $\equiv_{\mathcal{G}}$ for the class~$\mathcal{G}$ of all graphs is in quasipolynomial time. 
These examples also show that, even for \enquote{natural} graph classes, there is no apparent relation between the complexities of $\equiv_{\mathcal{T}}$, $\equiv_{\mathcal{P}}$, and $\equiv_{\mathcal{G}}$, despite that $\mathcal{T} \subseteq \mathcal{P} \subseteq \mathcal{G}$.
Making sense of this counterintuitive behavior is the objective of a line of work that studies the complexity of
homomorphism indistinguishability relations~$\equiv_{\mathcal{F}}$ in terms of graph-theoretic or model-theoretic properties of the graph class~$\mathcal{F}$ \cite{boker_complexity_2019,seppelt2024algorithmic,cerny_homomorphism_2026}.

Except for pathological\footnote{\label{footnote:pathological}Restricting wlog to homomorphism distinguishing closed graph classes \cite{roberson_oddomorphisms_2022}, the known dense~$\mathcal{F}$ with decidable~$\equiv_{\mathcal{F}}$ are the classes of all graphs and of all bipartite graphs, for which homomorphism indistinguishability is graph-isomorphism-complete \cite{lovasz_operations_1967,dvorak_2010_homomorphisms,seppelt_homomorphism_2024}, and the class of all complete graphs \cite{boker_complexity_2019}, for which the problem is $\mathsf{C}_=\mathsf{P}$-complete.} cases, all known examples for graph classes $\mathcal{F}$ with decidable $\equiv_{\mathcal{F}}$ are \emph{sparse}, in the sense that they exclude large bicliques as subgraphs \cite{nesetril2014sparsity}.
In this work, we give the first general algorithm for deciding homomorphism indistinguishability over dense graph classes.

\begin{theorem}\label{thm:cw-main}
    For $k \geq 1$,
    homomorphism indistinguishability over the class of graphs of cliquewidth~$\leq k$ can be decided in $\mathsf{EXPTIME}$.
\end{theorem}

Cliquewidth is the dense analogue of treewidth, in the sense that these parameters are functionally equivalent for $K_{t,t}$-subgraph-free graphs \cite{goos_tree-width_2000,courcelle_upper_2000}.
Our algorithm for \cref{thm:cw-main} is inspired by the $k$-dimensional Weisfeiler--Leman algorithm ($k$-WL), which decides homomorphism indistinguishability over the class of graphs of treewidth~$\leq k$ \cite{dvorak_2010_homomorphisms}.
The $k$-WL algorithm is one of the most prominent graph isomorphism heuristics due to its deep connections to finite model theory \cite{cai_furer_immerman_1992,immerman1990describing}, proof complexity \cite{BerkholzG15,RezendeFJN025,ToranW24}, graph machine learning \cite{morris_weisfeiler_2019,MorrisLMRKGFB22,ShervashidzeSLMB11,xu_how_2019}, combinatorics \cite{dvorak_2010_homomorphisms,dell_grohe_rattan_2018}, and mathematical programming \cite{atserias_sherali-adams_2012,grohe_pebble_2015}.

The ($1$-dimensional) Weisfeiler--Leman algorithm iteratively colors the vertices of the input graph.
In each round, it aggregates the multiset of colors of the neighbors of a given vertex and encodes the resulting color histogram in the new color of this vertex.
While the Weisfeiler--Leman algorithm fails to distinguish e.g.\ non-isomorphic $d$-regular graphs, it can be implemented to run in time $O((n+m)\log n)$ on $n$-vertex $m$-edge input graphs and is part of all competitive graph isomorphism solvers \cite{McKay81,McKayP14,JunttilaK07,JunttilaK11,AndersS21}.
Its $k$-dimensional generalization, which colors $k$-tuples of vertices,
runs in time $O(n^{k+1}\log n)$ and
is invoked by Babai's algorithm \cite{babai_graph_2016}.

We design a dense analogue of the Weisfeiler--Leman algorithm for deciding homomorphism indistinguishability over the class of graphs of cliquewidth~$\leq k$.
While $k$-WL colors $k$-tuples of vertices, our $k$-dimensional dense Weisfeiler--Leman algorithm colors $k$-tuples of subsets of the vertex set. 
Both algorithms work by aggregating colors of incident tuples and thereby refining the previously computed coloring until this process stabilizes.
A detailed definition is given in \cref{sec:dense-weisfeiler--leman-algorithm}.
Since our algorithm colors an exponentially large set, it runs in exponential time.
The complexity jump from $\mathsf{PTIME}$ \cite{grohe_equivalence_1999} to $\mathsf{EXPTIME}$ is expected in light of similar phenomena for decision problems with succinct encodings~\cite{papadimitriou_note_1986}.
We achieve the precise runtime in \cref{thm:main-runtime} by employing fast M\"obius and Zeta transforms \cite{bjorklund_fourier_2007}.
To our knowledge, this is the first application of an advanced exponential-time algorithmic technique in the context of homomorphism indistinguishability.

\begin{theorem}[label=thm:main-runtime,restate=thmRuntime]
Let $k\geq1$. On an $n$-vertex input graph, the $k$-dimensional dense
Weisfeiler--Leman algorithm runs in time
$4^{kn}\cdot\poly(n,k)$.
\end{theorem}

The Weisfeiler--Leman algorithm is not only studied because of its efficiency but also due to the various characterizations of its distinguishing power.
Paralleling the characterization of $k$-WL indistinguishability in terms of first-order logic with counting quantifiers,
we prove a logical characterization of the dense WL algorithm by introducing the counting logic~$\mathsf{CW}$.
Also, as indicated above, the dense WL algorithm is characterized by homomorphism indistinguishability over the class of graphs of cliquewidth~$\leq k$.
This is the first characterization of a homomorphism indistinguishability relation over a dense graph class.

\begin{theorem}\label{thm:main-characterisation}
    Let $k \geq 1$.
    For graphs $G$ and $H$, the following are equivalent:
    \begin{enumerate}
        \item $G$ and $H$ are homomorphism indistinguishable over all graphs of cliquewidth~$\leq k$,
        \item $G$ and $H$ are not distinguished by the $k$-dimensional dense Weisfeiler--Leman algorithm,
        \item $G$ and $H$ are $\mathsf{CW}^k$-equivalent.
    \end{enumerate}
\end{theorem}

We use the logical characterization to compare the distinguishing power of $k$-WL and the dense WL algorithm in \cref{thm:dense-to-sparse}, generalizing a result of \textcite{dvorak_2010_homomorphisms}.
In particular, we show that the dense WL algorithm is strictly more powerful than WL in the sense that, for every $k \geq 1$, there exist graphs $G$ and $H$ that are not distinguished by $k$-WL but by the $2$-dimensional dense WL algorithm (\cref{cor:wl-dwl-2-to-infty}).
On the other hand, if $G$ and $H$ are assumed to be $K_{t,t}$-subgraph-free, then dense WL and WL have the same distinguishing power up to a constant offset in dimension.

\Cref{thm:cw-main} is a step towards mapping out the complexity of homomorphism indistinguishability relations~$\equiv_{\mathcal{F}}$ in terms of graph- and model-theoretic properties of the graph class~$\mathcal{F}$. 
In this context, \textcite{seppelt2024algorithmic} showed that $\equiv_{\mathcal{F}}$ can be decided in randomized polynomial time~$\mathsf{coRP}$ for every bounded-treewidth graph class that is definable in monadic second-order logic $\mathsf{CMSO}_2$.
Moreover, \textcite{cerny_homomorphism_2026} showed that $\equiv_{\mathcal{F}}$ is in $\mathsf{C}_=\mathsf{L} \subseteq \mathsf{NC}^2$ for every bounded-pathwidth $\mathsf{CMSO}_2$-definable graph class.
These results imply state-of-the-art algorithms for semidefinite programming relaxations of graph and quantum isomorphism \cite{kar_npa_2026,seppelt2024algorithmic}.
Building on the machinery of~\cite{cerny_homomorphism_2026},
we show that homomorphism indistinguishability over every $\mathsf{CMSO}_1$-definable graph class of bounded cliquewidth is decidable.
Up to a minor discrepancy between $\mathsf{CMSO}_1$ and $\mathsf{CMSO}_2$,
\enquote{$\mathsf{CMSO}_1$-definable bounded-cliquewidth} is the most general sufficient condition on a non-trivial graph class $\mathcal{F}$ for $\equiv_{\mathcal{F}}$ to be decidable, subsuming the previously most general condition from~\cite{seppelt2024algorithmic}, see \cref{footnote:pathological}.

\begin{theorem}[restate=mainCWMSO, label=thm:main-clique-width-mso]
    For every $\mathsf{CMSO}_1$-definable graph class $\mathcal{F}$ of bounded cliquewidth, homomorphism indistinguishability over $\mathcal{F}$ can be decided in randomized exponential time with one-sided error.
\end{theorem}
Moreover, we consider graph classes of bounded linear cliquewidth, which is the dense analogue of bounded pathwidth.
\begin{theorem}[restate=mainLCWMSO, label=thm:main-linear-clique-width-mso]
    For every $\mathsf{CMSO}_1$-definable graph class $\mathcal{F}$ of bounded linear cliquewidth, homomorphism indistinguishability over $\mathcal{F}$ can be decided in $\mathsf{PSPACE}$.
\end{theorem}

The proofs of both \cref{thm:main-clique-width-mso,thm:main-linear-clique-width-mso} reduce the homomorphism indistinguishability problem to equivalence testing for succinctly given multiplicity automata with exponentially many states. Therefore, the complexity jumps from randomized polynomial time to randomized exponential time and from $\mathsf{NC}$ to $\mathsf{PSPACE}$ are as predicted in~\cite{papadimitriou_note_1986}, see \cref{fig:complexity-overview}.
In the case of \cref{thm:main-linear-clique-width-mso}, we show a tight lower bound when considering homomorphisms
between vertex-colored graphs (despite being stated for uncolored graphs, the algorithms in \cref{thm:main-linear-clique-width-mso,thm:main-clique-width-mso} also work for vertex-colored graphs).
The previous best lower bound in this regime was $\mathsf{C}_=\mathsf{P}$\nobreakdash-completeness \cite{boker_complexity_2019} for homomorphism indistinguishability over the class of all complete graphs, which is $\mathsf{FO}$\nobreakdash-definable and of bounded linear cliquewidth.

\begin{theorem} \label{thm:pspace_hard}
    There is a fixed finite color set $\Gamma$ and a fixed
    $\mathsf{MSO}_1$-definable class $\mathcal F$ of $\Gamma$-vertex-colored graphs of bounded
    linear cliquewidth such that homomorphism indistinguishability over $\mathcal F$ is $\mathsf{PSPACE}$-hard.
\end{theorem}

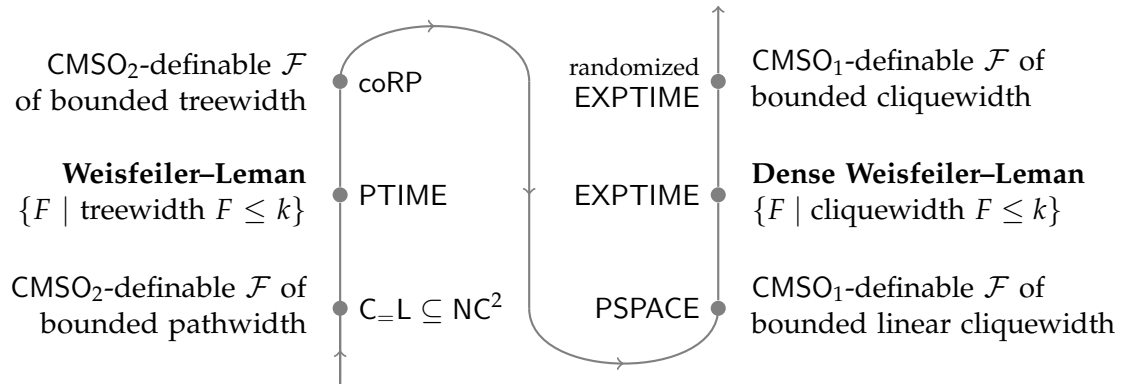
\begin{figure}
    \centering
    \begin{tikzpicture}
        [cclass/.style={fill=gray,circle, inner sep=2pt},
        decoration={
    markings,
    mark=at position 0.5 with {\arrow{>}}}]

        \node [cclass] (CL) {};
        \node [anchor=west] at (CL.east) {$\mathsf{C}_=\mathsf{L} \subseteq \mathsf{NC}^2$};

        \coordinate [below of=CL] (start);

        \node [cclass, above of=CL, yshift=.5cm] (PTIME) {};
        \node [anchor=west] at (PTIME.east) {$\mathsf{PTIME}$};

        \node [cclass, above of=PTIME, yshift=.5cm] (coRP) {};
        \node [anchor=west] at (coRP.east) {$\mathsf{coRP}$};

        \coordinate [above of=coRP, yshift=-1cm] (turn1);

        \coordinate [right of=turn1, xshift=1.5cm] (turn2);
        \coordinate [right of=start, xshift=1.5cm, yshift=1cm] (turn3);
        \coordinate [right of=turn3, xshift=1.5cm] (turn4);
        
        \draw [draw=gray, thick] (start) edge [postaction={decorate}] (CL) -- (PTIME) -- (coRP) -- (turn1);
        
        \draw [draw=gray, thick] (turn1) edge [bend right, out=90, in=90,postaction={decorate}] (turn2);
        \draw [draw=gray, thick,postaction={decorate}] (turn2) -- (turn3);
        \draw [draw=gray, thick] (turn3) edge [bend left=10, out=270, in=270,postaction={decorate}] (turn4);

        \node [cclass, right of=CL, xshift=4cm] (CL2) {};
        \node [anchor=east] at (CL2.west) {$\mathsf{PSPACE}$};

        \node [cclass, above of=CL2, yshift=.5cm] (PTIME2) {};
        \node [anchor=east] at (PTIME2.west) {$\mathsf{EXPTIME}$};

        \node [cclass, above of=PTIME2, yshift=.5cm] (coRP2) {};
        \node [anchor=east,yshift=.2cm,font=\smaller] at (coRP2.west) {randomized};
        \node [anchor=east,yshift=-.2cm] at (coRP2.west) {$\mathsf{EXPTIME}$};

        \draw [draw=gray, thick] (turn4) -- (CL2) -- (PTIME2) -- (coRP2);

        \coordinate [above of=coRP2] (end) ;
        \draw [draw=gray, thick,->] (coRP2) -- (end);

        \node [text width=4cm, anchor=east, align=right,xshift=-.2cm] at (CL.west) {$\mathsf{CMSO}_2$-definable $\mathcal{F}$ of  bounded pathwidth};

        \node [text width=4cm, anchor=east, align=right,xshift=-.2cm] at (PTIME.west) {\textbf{Weisfeiler--Leman}\\ $\{F \mid \text{treewidth } F \leq k\}$};

        \node [text width=4cm, anchor=east, align=right,xshift=-.2cm] at (coRP.west) {$\mathsf{CMSO}_2$-definable $\mathcal{F}$ of bounded treewidth};

        \node [text width=5cm, anchor=west, align=left,xshift=.2cm] at (CL2.east) {$\mathsf{CMSO}_1$-definable $\mathcal{F}$ of \\ bounded linear cliquewidth};

        \node [text width=5cm, anchor=west, align=left,xshift=.2cm] at (PTIME2.east) {\textbf{Dense Weisfeiler--Leman}\\$\{F \mid \text{cliquewidth } F \leq k\}$};

        \node [text width=5cm, anchor=west, align=left,xshift=.2cm] at (coRP2.east) {$\mathsf{CMSO}_1$-definable $\mathcal{F}$ of \\  bounded cliquewidth};
    \end{tikzpicture}
    \caption{Comparison of complexity results of \cite{cerny_homomorphism_2026,grohe_equivalence_1999,immerman1990describing} for sparse graph classes (left column) with our \cref{thm:cw-main,thm:main-clique-width-mso,thm:main-linear-clique-width-mso} for dense graph classes (right column).}
    \label{fig:complexity-overview}
\end{figure}

Finally, we show that the $2$-dimensional dense Weisfeiler--Leman algorithm cannot be executed in polynomial time unless $\mathsf{FPT} = \mathsf{W[1]}$. 
This is a consequence of the following theorem and \cref{thm:main-characterisation} 
noting that \emph{cographs} are precisely the graphs of cliquewidth~$\leq 2$.
Curiously, our reduction does not immediately show hardness for $k$-dimensional dense WL for any $k \neq 2$.

\begin{theorem}[label=thm:cographs,restate=thmCographs]
    Unless $\mathsf{FPT} = \mathsf{W[1]}$, there is no polynomial-time algorithm for deciding homomorphism indistinguishability over all cographs.
\end{theorem}

\section{Preliminaries}

\paragraph{Tuples and indices.}
For any set $X$ and $k \geq 1$, we denote tuples in $X^k$ by boldface letters $\boldsymbol{x}$ and write $x_1, \dots, x_k$ for the entries.
For an integer $\ell \in [k]$ and $x \in X$,
write $\boldsymbol{x}[\ell/x]$ for the tuple obtained from $\boldsymbol{x}$ by replacing the $\ell$-th entry by $x$.
For $\ell_1, \dots, \ell_n \in [k]$ and $x_1, \dots, x_n \in X$, write $\boldsymbol{x}[\ell_1/x_1, \dots, \ell_n/x_n]$ for $\boldsymbol{x}[\ell_1/x_1][\ell_2/x_2] \dots [\ell_n/x_n]$.

\paragraph{Graphs and homomorphisms.}
All graphs are finite, undirected, without loops and multiple edges.
For a graph $G$, write $P(G)$ for the powerset of the vertex set~$V(G)$.
A \emph{homomorphism} from a graph $F$ to a graph $G$ is a map $h \colon V(F) \to V(G)$ such that $h(u)h(v) \in E(G)$ whenever $uv \in E(F)$.
We write $\hom(F, G)$ for the number of homomorphisms from $F$ to $G$.

\paragraph{Cliquewidth.}
Fix $k\ge 1$.
A \emph{$k$-partitioned graph} is a tuple $\boldsymbol{F}=(F;C_1,\dots,C_k)$ of a graph $F$ together with a partition
$V(F)=C_1\uplus \cdots \uplus C_k$ of its vertex set.
The vertices in the set $C_i$ are referred to as the \emph{vertices of color~$i$}.
For $i \in [k]$, we write $\bullet_i$ for the $k$-partitioned graph with a single vertex of color~$i$.
A $k$-partitioned graph has \emph{cliquewidth~$\leq k$} if it can be constructed from $\bullet_1, \dots, \bullet_k$ by the following operations:
For $k$-partitioned graphs $\boldsymbol{F}_1, \boldsymbol{F}_2,\boldsymbol{F}$,
\begin{itemize}
  \item $\boldsymbol{F}_1 \oplus \boldsymbol{F}_2$ denotes the disjoint union of $\boldsymbol{F}_1$ and $\boldsymbol{F}_2$;
  \item $\boldsymbol{F}^{\eta(i,j)}$ for $i\neq j$ is constructed from $\boldsymbol{F}$ by adding all edges between every vertex of color $i$ and every vertex of color $j$;
  \item $\boldsymbol{F}^{\rho(i \to j)}$ for $i\neq j$ is constructed from $\boldsymbol{F}$ by recoloring all vertices of color $i$ to color $j$.
\end{itemize}
We write $\mathfrak{A}_k$ for the algebra (in the sense of universal algebra) of $k$-partitioned graphs of cliquewidth $\leq k$ with the constants $\bullet_i$ for $i \in [k]$,
the binary operation~$\oplus$, and the unary operations $\eta(i,j)$ and $\rho(i \to j)$ for $i \neq j$.
For background on this algebra, see \cite{courcelle_graph_2012}.

A $k$-partitioned graph has \emph{linear cliquewidth~$\leq k$} if it can be constructed from $\bullet_1, \dots, \bullet_k$ by the unary operations $\oplus {\bullet_i}$ for $i \in [k]$, $\eta(i,j)$, and $\rho(i \to j)$.
That is, arbitrary disjoint unions are forbidden; vertices have to be added one at a time.
For example, an $n$-vertex complete graph can be constructed via $\boldsymbol{K}_1 \coloneqq \bullet_1$ and $\boldsymbol{K}_{n+1} \coloneqq ((\boldsymbol{K}_{n} \oplus \bullet_2)^{\eta(1,2)})^{\rho(2 \to 1)}$ for $n \geq 1$ and thus has linear cliquewidth~$\leq 2$.

\paragraph{Multiplicity automata.}
A \emph{multiplicity word automaton} (\MWA, \cite{schutzenberger_definition_1961}) is a tuple $\mathcal{A} = (S, \Sigma, M, \alpha, \gamma)$ of a finite set $S$ of \emph{states},
a finite alphabet~$\Sigma$, a map $M \colon \Sigma \to \mathbb{Q}^{S\times S}$ associating to a letter $\sigma \in \Sigma$ a \emph{transition matrix} $M(\sigma) \in \mathbb{Q}^{S\times S}$, and \emph{initial} and \emph{final} vectors $\alpha, \gamma \in \mathbb{Q}^S$. Its semantics are given by the function $\den{\mathcal{A}} \colon \Sigma^* \to \mathbb{Q}$ mapping a word $w = w_1 \cdots w_n$ to  $\alpha^\top M(w_1) \cdots M(w_n) \gamma \in \mathbb{Q}$.
Two \MWAs $\mathcal{A}$ and $\mathcal{B}$ are \emph{equivalent} if $\den{\mathcal{A}} = \den{\mathcal{B}}$.
Deciding equivalence of \MWAs is $\mathsf{C}_=\mathsf{L}$-complete \cite{cerny_homomorphism_2026}.

A~\emph{ranked alphabet} is a set $\Omega$ of symbols with 
assignments of \emph{arities}~$|\sigma| \in \mathbb{N}$ to each symbol $\sigma \in \Omega$.
For $n\in \mathbb N$, the set of all $n$-ary symbols in $\Omega$ is denoted by $\Omega_n$.
The set of \emph{$\Omega$-trees}, denoted by $T_\Omega$, is the smallest set such that $\Omega_0 \subseteq T_\Omega$,
and if $n\ge 1$, $\sigma\in \Omega_n$, $t_1, \dots, t_n \in T_\Omega$, then  $\sigma(t_1, \dots, t_n)\in T_\Omega$.

A \emph{multiplicity tree automaton} (\MTA, \cite{berstel_recognizable_1982}) is a~tuple $\mathcal{A} = (S, \Omega, \mu, \gamma)$,
where $S$ is a finite set of states, $\Omega$ is a~ranked alphabet, 
$\mu$ is 
a \emph{tree representation}, i.e.\ a collection of maps $M_n\colon \Omega_n \to \mathbb{Q}^{S^n \times S}$ for each arity~$n$
associating to a symbol $\sigma \in \Omega_n$
a \emph{transition matrix} $\mu(\sigma) \coloneqq M_n(\sigma) \in \mathbb{Q}^{S^n \times S}$,
and $\gamma\in \mathbb{Q}^S$, the \emph{final vector}.
The semantics of an \MTA $\mathcal{A}$ are defined by extending the tree representation $\mu$ from $\Omega$ to all elements $\sigma(t_1, \dots, t_n) \in T_\Omega$ by
\(
    \mu (\sigma(t_1, \dots, t_n)) \coloneqq \left(\mu(t_1) \otimes \cdots \otimes \mu(t_n)\right) \cdot \mu(\sigma). 
\)
Here, $\otimes$ denotes the Kronecker product of matrices.
Finally, $\den{\mathcal{A}}(t) \coloneqq \mu(t) \cdot \gamma \in \mathbb{Q}$.
Two \MTAs $\mathcal{A}$ and $\mathcal{B}$ are \emph{equivalent}
if $\den{\mathcal{A}} = \den{\mathcal{B}}$.
Deciding equivalence of \MTAs is logspace many-one interreducible with polynomial identity testing \cite{marusic_complexity_2015}, which is in randomized polynomial time $\mathsf{coRP}$ by the Schwartz--Zippel algorithm \cite{schwartz_fast_1980,zippel_probabilistic_1979}.

\paragraph{Monadic second-order logic.}
The logics $\mathsf{CMSO}_1$ and $\mathsf{CMSO}_2$ are counting extensions of
monadic second-order logic over graphs. In $\mathsf{MSO}_1$, a graph is given
by its vertices and the adjacency relation, and quantification ranges over
vertices and vertex sets. In $\mathsf{MSO}_2$, graphs are given by their vertex set, edge set, and the incidence relation; it allows quantification both over vertices and vertex sets as well as over edges and edge sets, making it strictly more expressive than $\mathsf{MSO}_1$.
Then $\mathsf{CMSO}_1$ and $\mathsf{CMSO}_2$ extend $\mathsf{MSO}_1$ and $\mathsf{MSO}_2$, respectively, with modular counting predicates
$\operatorname{card}_{p,q}(X)$, which hold when $|X| \equiv p \pmod{q}$. By Courcelle's theorem \cite{courcelle_graph_2012}, $\mathsf{CMSO}_2$-properties
are decidable in linear time on graphs of bounded treewidth, and
$\mathsf{CMSO}_1$-properties on graphs of bounded cliquewidth, see \cref{thm:courcelle}.

\section{Counting homomorphisms from graphs of bounded cliquewidth}

In this section, we construct our main tool: a representation of the algebra $\mathfrak{A}_k$ of $k$-partitioned graphs of cliquewidth $\leq k$ with the operations disjoint union, adding bicliques, and recoloring. 
This representation encodes homomorphism counts from bounded-cliquewidth graphs into a fixed graph~$G$.

The corresponding construction for the algebra of $k$-labeled graphs of treewidth~$< k$ underpins the standard dynamic programming approach for counting homomorphisms from graphs of bounded treewidth \cite[Example~3.4]{flum_parameterized_2004}. 
Furthermore, it has been used to give linear-algebraic characterizations and algorithms for homomorphism indistinguishability over bounded-treewidth graph classes \cite{grohe_homomorphism_2025,rattan_weisfeiler-leman_2023,dell_grohe_rattan_2018,cerny_homomorphism_2026,seppelt2024algorithmic}.
To our knowledge, no such construction was known for cliquewidth.

More precisely, we construct, for every graph $G$, 
a representation \(\mathfrak{A}_k \to \mathbb{Q}^{P(G)^k}\) (recall that $P(G)$ is the powerset of $V(G)$)
of the algebra $\mathfrak{A}_k$, i.e., a map associating $k$-partitioned graphs from~$\mathfrak{A}_k$ with vectors in $\mathbb{Q}^{P(G)^k}$ and (bi)linear operations on $\mathbb{Q}^{P(G)^k}$ corresponding to each of the operations of $\mathfrak{A}_k$.
Formally, the representation is given by
$
    \boldsymbol{F} \mapsto \boldsymbol{F}_G
$
where $\boldsymbol{F}_G \in \mathbb{Q}^{P(G)^k}$ is the vector defined  for $\boldsymbol{D} \in P(G)^k$ as
\begin{equation}\label{def:hom-vector}
\boldsymbol{F}_G(\boldsymbol{D})
\coloneqq \left| \left\{ h \colon F \to G \mid h(C_i) \subseteq D_i \text{ for all } i\in [k] \right\} \right|.
\end{equation}
That is, $\boldsymbol{F}_G(\boldsymbol{D})$ is equal to the number of homomorphisms from $F$ to $G$ that map $C_i$ into $D_i$ for $i \in [k]$.
We call $\boldsymbol{F}_G$ a \emph{homomorphism vector}.
We first compute these values for the single-vertex graphs~$\bullet_i$.

\begin{example}\label{ex:bullet}
    Let $i \in [k]$ and $\boldsymbol{F} \coloneqq \bullet_i$.
    For every graph $G$ and $\boldsymbol{D} \in P(G)^k$,
    it holds that
    $\boldsymbol{F}_G(\boldsymbol{D}) = |D_i|$.
\end{example}
At the entry $(V(G), \dots, V(G))$,
the vector $\boldsymbol{F}_G$ contains the number of homomorphisms $F\to G$.
\begin{observation}\label{obs:drop-color}
    For every graph $G$ and every $k$-partitioned graph $\boldsymbol{F}$, it holds that
    \[
        \boldsymbol{F}_G(V(G), \dots, V(G)) = \hom(F, G).
    \]
\end{observation}

Our central lemma shows that the cliquewidth operations $\oplus$, $\eta(i,j)$, and $\rho(i \to j)$ can be linearly represented.
Here, recoloring and disjoint union are treated easily.
The crux lies in the operation~$\eta(i,j)$, which we handle by applying M\"obius and Zeta transforms on the subset lattice.

\begin{lemma}\label{lem:identities}
    Let $k \geq 1$.
    Let $G$ be a graph and $\boldsymbol{D} \in P(G)^k$.
    The following hold for all $k$-partitioned graphs $\boldsymbol{F}, \boldsymbol{F}'$ and $1 \leq i \neq j \leq k$.
    \begin{align}
        (\boldsymbol{F} \oplus \boldsymbol{F}')_G(\boldsymbol{D})
        &= \boldsymbol{F}_G(\boldsymbol{D}) \cdot \boldsymbol{F}'_G(\boldsymbol{D}) ,\label{eq:disjoint-union} \\
        (\boldsymbol{F}^{\rho(i \to j)})_G(\boldsymbol{D}) &= \boldsymbol{F}_G(\boldsymbol{D}[i/D_j]), \label{eq:rename}\\
        (\boldsymbol{F}^{\eta(i,j)})_G(\boldsymbol{D}) &=
        \sum_{\substack{D''_i \subseteq D'_i \subseteq D_i, \\ D''_j \subseteq D'_j \subseteq D_j}} (-1)^{|D'_i \setminus D''_i|+|D'_j \setminus D''_j|} \cdot 
        \beta_G(D'_i, D'_j) \cdot \boldsymbol{F}_G(\boldsymbol{D}[i/D''_i, j/D''_j]),\label{eq:join}
    \end{align}
    where $ \beta_G(A, B)$ is $1$ if $A \cap B = \emptyset$ and the bipartite graph $G[A, B]$ is complete, and $0$ otherwise.
\end{lemma}

The precise identities in \cref{lem:identities}
are less relevant than the consequence that all three operations correspond to (bi)linear operations on $\mathbb{Q}^{P(G)^k}$.
Indeed, the maps $\boldsymbol{F}_G \mapsto (\boldsymbol{F}^{\rho(i\to j)})_G$ and
$\boldsymbol{F}_G \mapsto (\boldsymbol{F}^{\eta(i,j)})_G$
extend to linear maps $R^{ij}_G, E^{ij}_G \colon \mathbb{Q}^{P(G)^k} \to \mathbb{Q}^{P(G)^k}$ defined
for $v \in \mathbb{Q}^{P(G)^k}$ and $\boldsymbol{D} \in P(G)^k$ via
\begin{align}
    (R_G^{ij} v)(\boldsymbol{D}) &\coloneqq v(\boldsymbol{D}[i/D_j]), \label{eq:Rij}\\
    (E_G^{ij} v)(\boldsymbol{D}) &\coloneqq \sum_{\substack{D''_i \subseteq D'_i \subseteq D_i, \\ D''_j \subseteq D'_j \subseteq D_j}} (-1)^{|D'_i \setminus D''_i|+|D'_j \setminus D''_j|} \cdot 
        \beta_G(D'_i, D'_j) \cdot v(\boldsymbol{D}[i/D''_i, j/D''_j]).\label{eq:Eij}
\end{align}
Similarly, the map $(\boldsymbol{F}_G, \boldsymbol{F}'_G) \mapsto (\boldsymbol{F} \oplus \boldsymbol{F}')_G$
 extends to a bilinear map $\mathbb{Q}^{P(G)^k} \times \mathbb{Q}^{P(G)^k} \to \mathbb{Q}^{P(G)^k}$.

\begin{proof}[Proof of \cref{lem:identities}]
    Write $\boldsymbol{F} = (F; C_1, \dots, C_k)$
    and $\boldsymbol{F}' = (F'; C'_1, \dots, C'_k)$.
    For \cref{eq:disjoint-union},
    observe that the homomorphisms $h \colon F \oplus F' \to G$ such that $h(C_i \cup C'_i) \subseteq D_i$ for $i \in [k]$ are in bijection with pairs of homomorphisms $h \colon F \to G$ and $h' \colon F' \to G$ such that $h(C_i) \subseteq D_i$ and $h'(C'_i) \subseteq D_i$ for all $i \in [k]$.

    For \cref{eq:rename},
    note that the homomorphisms $h \colon F \to G$ such that $h(C_\ell) \subseteq D_\ell$ for all $\ell \in [k] \setminus \{i,j\}$
    and
    $h(C_i \cup C_j) \subseteq D_j$
    are precisely the homomorphisms $h \colon F\to G$ such that $h(C_\ell) \subseteq D_\ell$ for all $\ell \in [k] \setminus \{i,j\}$
    and $h(C_i) \subseteq D_j$ and $h(C_j) \subseteq D_j$.

    For \cref{eq:join},
    we apply M\"obius inversion on the subset lattice to relate homomorphism counts and surjective homomorphism counts. 
    Write $\surj_{i,j}(\boldsymbol{F}, (G; D_1, \dots, D_k))$ 
    for the number of homomorphisms $h \colon F \to G$ such that $h(C_\ell) \subseteq D_\ell$ for all $\ell \in [k]$ and $h(C_i) = D_i$ and $h(C_j) = D_j$.
    Counting surjective homomorphisms has the advantage that the join operation $\eta(i,j)$ can be  handled easily. Indeed, the following holds:
    \begin{claim}\label{claim:surj-biclique}
        $\surj_{i,j}(\boldsymbol{F}^{\eta(i,j)}, (G; \boldsymbol{D}))
        = \beta_G(D_i, D_j) 
    \cdot \surj_{i,j}(\boldsymbol{F}, (G; \boldsymbol{D})).$
    \end{claim}
    \begin{claimproof}
If $C_i = \emptyset$ or $C_j = \emptyset$, then $\eta(i,j)$ adds no edges, so $\boldsymbol{F}^{\eta(i,j)} = \boldsymbol{F}$ and both sides are equal: the surjectivity condition $h(C_i) = D_i$ forces $D_i = \emptyset$ when $C_i = \emptyset$, and $\beta_G(\emptyset, D_j) = 1$ holds vacuously.

Now assume $C_i \neq \emptyset$ and $C_j \neq \emptyset$.
Assume first that $D_i\cap D_j=\emptyset$ and $G[D_i,D_j]$ is complete bipartite.
Let $h$ be counted by $\surj_{i,j}(\boldsymbol{F},(G;D_1,\dots,D_k))$.
Then $h(C_i)=D_i$ and $h(C_j)=D_j$.
For every newly added edge $uv$ with $u\in C_i$ and $v\in C_j$, we have $h(u)\in D_i$ and $h(v)\in D_j$, and hence $h(u)h(v)\in E(G)$.
Thus $h$ is also counted by $\surj_{i,j}(\boldsymbol{F}^{\eta(i,j)},(G;D_1,\dots,D_k))$.

Conversely, suppose either $D_i\cap D_j\neq\emptyset$ or $G[D_i,D_j]$ is not complete bipartite.
If we have $\surj_{i,j}(\boldsymbol{F}^{\eta(i,j)},(G;D_1,\dots,D_k))>0$, let $h$ be a homomorphism counted there.
Then $h(C_i)=D_i$ and $h(C_j)=D_j$.
In particular, $D_i \neq \emptyset$ and $D_j \neq \emptyset$ since $C_i \neq \emptyset$ and $C_j \neq \emptyset$.
If $G[D_i,D_j]$ is not complete bipartite, choose $x\in D_i$ and $y\in D_j$ with $xy\notin E(G)$.
Since $h(C_i)=D_i$ and $h(C_j)=D_j$, there are vertices $u\in C_i$ and $v\in C_j$ with $h(u)=x$ and $h(v)=y$.
But $uv$ is an edge of $\boldsymbol{F}^{\eta(i,j)}$, contradiction.
If $D_i\cap D_j\neq\emptyset$, choose $x\in D_i\cap D_j$.
Since $h(C_i)=D_i$ and $h(C_j)=D_j$, there exist $u\in C_i$ and $v\in C_j$ with $h(u)=h(v)=x$.
But $uv$ is an edge of $\boldsymbol{F}^{\eta(i,j)}$, so $h(u)h(v)=xx$ must be an edge of $G$, contradicting that $G$ is loopless.
Thus no homomorphism is counted in this case.
    \end{claimproof}

    It remains to relate surjective homomorphism counts to homomorphism counts via M\"obius inversion. 
    It holds that
    \begin{equation}\label{eq:hom-in-surj}
        \boldsymbol{F}_G(\boldsymbol{D})
        = \sum_{\substack{D'_i \subseteq D_i \\ D'_j \subseteq D_j}}
        \surj_{i,j}(\boldsymbol{F}, (G; \boldsymbol{D}[i/D'_i, j/D'_j]))
    \end{equation}
    and hence, by M\"obius inversion,
    \begin{equation}\label{eq:surj-in-hom}
        \surj_{i,j}(\boldsymbol{F},(G; \boldsymbol{D}))
        = \sum_{\substack{D'_i \subseteq D_i \\ D'_j \subseteq D_j}}
        (-1)^{|D_i \setminus D'_i| + |D_j \setminus D'_j|}
        \boldsymbol{F}_G(\boldsymbol{D}[i/D'_i, j/D'_j]).
    \end{equation}
    Combining \cref{claim:surj-biclique,eq:hom-in-surj,eq:surj-in-hom} yields \cref{eq:join}.
\end{proof} 
 
\section{Dense Weisfeiler--Leman algorithm}
\label{sec:dense-weisfeiler--leman-algorithm}
Equipped with \cref{lem:identities},
we now formally introduce our dense Weisfeiler--Leman algorithm.
To that end, recall the classical $1$-dimensional Weisfeiler--Leman algorithm.
It colors vertices of the input graph~$G$ starting from an initial monochromatic coloring~$\chi^0_G$ using iterative updates given by
\begin{equation}\label{eq:wl}
\chi^{r+1}_G(v) \coloneqq \left( \chi^r_G(v),\;  \multiset{\chi^r_G(w) \mid w \in N_G(v)} \right).
\end{equation}
Towards our dense Weisfeiler--Leman algorithm, we recast \eqref{eq:wl} in linear-algebraic language. 
In every round~$r \geq 0$, the algorithm computes a coloring $C_1^r \uplus \dots \uplus C_{\ell_r}^r = V(G)$ of the vertex set.
Abusing notation, we write $\chi^r_G = \{C_1^r, \dots, C_{\ell_r}^r\}$
for the set of colors.
Writing $\boldsymbol{1}_C$ for the indicator vector of a set $C \subseteq V(G)$,
the update step~\eqref{eq:wl} can be rephrased as
\[
    \chi^{r+1}_G(v) \coloneqq \left( \chi^r_G(v),\; \left( (A_G \boldsymbol{1}_C)(v) \mid C \in \chi^r_G \right) \right).
\]
Here, $A_G \in \{0,1\}^{V(G) \times V(G)}$ 
is the adjacency matrix of $G$.
For a set $C \subseteq V(G)$ and a vertex $v \in V(G)$,
the value
$(A_G \boldsymbol{1}_C)(v)$
is equal to the number of neighbors $w$ of $v$ in the set~$C$.
Thus, $\left( (A_G \boldsymbol{1}_C)(v) \mid C \in \chi^r_G \right)$ is the vector of color multiplicities of neighbors of~$v$.
The $k$-dimensional Weisfeiler--Leman algorithm can be stated similarly, with $k$-tuples of vertices in place of single vertices $v$.

Our dense Weisfeiler--Leman algorithm colors $k$-tuples of \emph{subsets of} vertices of the input, i.e., the elements of $P(G)^k$.
Initially, a tuple $\boldsymbol{D} \in P(G)^k$ is colored by the vector of sizes $(|D_1|, \dots, |D_k|)$:
\begin{equation}\label{eq:initial}
\dwl_{G, k}^0(\boldsymbol{D}) \coloneqq (|D_1|, \dots, |D_k|).
\end{equation}
In each update step, color multiplicities of related tuples are aggregated.
Here, the matrices $R^{ij}_G$ and $E^{ij}_G$ from \cref{eq:Rij,eq:Eij} play the role of the adjacency matrix~$A_G$.
The update step is given by
\begin{equation}\label{eq:refine}
\dwl_{G, k}^{r+1}(\boldsymbol{D})
= \left( \dwl_{G, k}^{r}(\boldsymbol{D}), \; \left((R^{ij}_G \boldsymbol{1}_C)(\boldsymbol{D}), \ (E^{ij}_G \boldsymbol{1}_C)(\boldsymbol{D}) \mid C \in \dwl_{G, k}^{r}, \ i \neq j \right) \right).
\end{equation}
Since the $\dwl$ algorithm produces a strictly refining sequence of colorings of $P(G)^k$, it terminates after at most $2^{nk}-1$ rounds.
We denote the \emph{stable coloring} by $\dwl_{G, k}^\infty$.
See \cref{fig:execution} for an example.
Two graphs $G$ and $H$ are \emph{not distinguished by the $k$-dimensional dense Weisfeiler--Leman algorithm} if
\[
    \dwl_{G, k}^\infty(V(G), \dots, V(G))
    = \dwl_{H, k}^\infty(V(H), \dots, V(H)).
\]
We show that the $k$-dimensional dense Weisfeiler--Leman algorithm can be executed in exponential time using fast M\"obius and Zeta transformations~\cite{bjorklund_fourier_2007}.

\thmRuntime*

\begin{proof}We use a worklist implementation of the refinement process, analogous to the
standard partition-refinement implementation of Weisfeiler--Leman refinement
of~\cite{immerman1990describing}.
The initial partition of $P(G)^k$ is given by the colors
\[
        \boldsymbol D=(D_1,\ldots,D_k)
        \longmapsto
        (|D_1|,\ldots,|D_k|).
\]
All initial classes are put into a worklist.  When a current class is split,
all classes created by the split are put into the worklist.  When a class
$C$ is removed from the worklist, we ignore it if it is no longer current;
otherwise, for every ordered pair $i\neq j$, we refine the current partition
by the two functions
\[
        (R_G^{ij}\boldsymbol{1}_C)(\boldsymbol D)
        \coloneqq
        \mathbf 1_C(\boldsymbol D[i/D_j])
\]
and
\[
        (E_G^{ij}\boldsymbol{1}_C)(\boldsymbol D)
        \coloneqq
        \sum_{\substack{D''_i \subseteq D'_i \subseteq D_i\\
                        D''_j \subseteq D'_j \subseteq D_j}}
        (-1)^{|D'_i\setminus D''_i|+|D'_j\setminus D''_j|}
        \beta_G(D'_i,D'_j)
        \mathbf 1_C(\boldsymbol D[i/D''_i,j/D''_j]).
\]
These are the recoloring and join update functions from
\eqref{eq:rename} and \eqref{eq:join}, applied linearly to
$\boldsymbol 1_C$.

Call a partition \emph{stable} if, for every class $C$ and every ordered pair
$i\neq j$, both $R_G^{ij}\boldsymbol{1}_C$ and $E_G^{ij}\boldsymbol{1}_C$ are constant on every
class of the partition.  The stable dense WL coloring is the coarsest stable
refinement of the initial partition.  The worklist algorithm terminates with
a stable partition, as every final class is processed after its last
creation.

Let $\dwl_{G, k}^\infty$ be the coarsest stable refinement of the initial
partition.  We show that the algorithm never splits a class of
$\dwl_{G, k}^\infty$.  Indeed, if a current class $C$ is a union of
$\dwl_{G, k}^\infty$-classes, then, by linearity in $\mathbf 1_C$, the
functions $R_G^{ij}\boldsymbol{1}_C$ and $E_G^{ij}\boldsymbol{1}_C$ are sums of the corresponding
functions for the $\dwl_{G, k}^\infty$-classes contained in $C$.  Each
summand is constant on every $\dwl_{G, k}^\infty$-class by stability of
$\dwl_{G, k}^\infty$, and hence so is the sum.  Thus no refinement step cuts
a class of $\dwl_{G, k}^\infty$.  Therefore the final partition is stable and
no finer than $\dwl_{G, k}^\infty$, so it is exactly $\dwl_{G, k}^\infty$.

It remains to discuss the running time.  
Let
\(
        N \coloneqq |P(G)^k|=2^{kn}.
\)
The classes created during the
execution form a refinement tree with at most $N$ leaves, and therefore
fewer than $2N$ current classes are ever processed.  Fix one processed class
$C$ and one ordered pair $i\neq j$.  The function $R_G^{ij}\boldsymbol{1}_C$ is
computed by scanning all tuples in $P(G)^k$, in time
$N\cdot\poly(n,k)$.

For $E_G^{ij}\boldsymbol{1}_C$, the sum above is computed by two alternating subset sums
in coordinates $i,j$, followed by entrywise multiplication by
$\beta_G$, followed by two ordinary subset sums in coordinates $i,j$.  A
subset sum or alternating subset sum in one coordinate of a function
$P(G)^k\to\mathbb Q$ is computed by fixing the other $k-1$ coordinates and
applying the standard fast subset transform on $P(G)$.  There are
$(2^n)^{k-1}$
choices for the fixed coordinates, and each transform on $P(G)$ costs
$O(n^22^n)$ arithmetic operations; see, e.g.,
\textcite{bjorklund_fourier_2007}. Hence one coordinate transform costs
$O(n^22^{kn})=O(n^2N)$.  The factor $\beta_G$ is evaluated entrywise in
$N\cdot\poly(n)$ time.  The intermediate integers have $O(n)$ bits, so
$E_G^{ij}\boldsymbol{1}_C$ is computed in $N\cdot\poly(n,k)$ bit operations.

Thus processing one class costs $N\cdot\poly(n,k)$ bit operations.  Since
fewer than $2N$ classes are processed, all update functions are computed in
$N^2\cdot\poly(n,k)$ time.  Refining the partition after each processed
class, by sorting the $N$ tuples according to their current color and the
newly computed values, costs the same total time.  Hence the stable coloring
is computed in
\[
        N^2\cdot\poly(n,k)
        =
        (2^{kn})^2\cdot\poly(n,k)
        =
        4^{kn}\cdot\poly(n,k).
\]
For $k=1$, there are no ordered pairs $i\neq j$, and the same bound is
immediate.
\end{proof}

Finally, we show that the power of the $k$-dimensional dense Weisfeiler--Leman algorithm is characterized by homomorphism indistinguishability over the class of graphs of cliquewidth~$\leq k$.
As a corollary of \cref{thm:dwl-hom,thm:main-runtime}, 
it follows that homomorphism indistinguishability over the class of graphs of cliquewidth~$\leq k$ can be decided in exponential time, i.e., \cref{thm:cw-main}.

\begin{theorem}\label{thm:dwl-hom}
    Let $k\ge 1$.
    Two graphs $G$ and $H$ are not distinguished by the $k$-dimensional dense Weisfeiler--Leman algorithm if, and only if, they are homomorphism indistinguishable over all graphs of cliquewidth~$\leq k$.
\end{theorem}

\begin{proof}We first show the forward implication.
    It is implied by \cref{obs:drop-color} and the following claim.
    \begin{claim}
        For $\boldsymbol{F} \in \mathfrak{A}_k$, 
        there exist coefficients $\alpha_C$ indexed by $\dwl_{G,k}^
        \infty$- and $\dwl_{H,k}^
        \infty$-color classes such that 
        \[
            \boldsymbol{F}_G = 
            \sum_C \alpha_C \boldsymbol{1}_{C \cap P(G)^k}, \quad\quad
            \boldsymbol{F}_H = 
            \sum_C \alpha_C \boldsymbol{1}_{C \cap P(H)^k}.
        \]
    \end{claim}

    Indeed, if $(V(G), \dots, V(G))$ and $(V(H), \dots, V(H))$ have the same color, then $\hom(F, G) = \boldsymbol{F}_G(V(G), \dots, V(G)) =\boldsymbol{F}_H(V(H), \dots, V(H)) = \hom(F, H)$, as desired.
    
    \begin{claimproof}
        The proof is by structural induction on $\boldsymbol{F}$.
        If $\boldsymbol{F}$ is a single-vertex graph~$\bullet_i$,
        then the claim follows from \cref{ex:bullet} by definition of the initial coloring~\eqref{eq:initial}.

        If $\boldsymbol{F} = \boldsymbol{F}^1 \oplus \boldsymbol{F}^2$ for two graphs $\boldsymbol{F}^1$ and $\boldsymbol{F}^2$, then the claim follows by induction and~\eqref{eq:disjoint-union}.
        Indeed, if $\alpha^1_C$ and $\alpha^2_C$ are the coefficients for $\boldsymbol{F}^1$ and $\boldsymbol{F}^2$, then
        \[
        \boldsymbol{F}_G = \boldsymbol{F}^1_G \odot \boldsymbol{F}^2_G
        = \left(\sum_{C_1} \alpha^1_{C_1} \boldsymbol{1}_{C_1 \cap P(G)^k}\right)
        \odot \left(\sum_{C_2} \alpha^2_{C_2} \boldsymbol{1}_{C_2 \cap P(G)^k}\right)
        = \sum_C \alpha^1_C \alpha^2_C \boldsymbol{1}_{C \cap P(G)^k}
        \]
        and analogously for $H$.
        Here, $\odot$ denotes the entrywise product.
    
        If $\boldsymbol{F} = \boldsymbol{K}^{\rho(i \to j)}$,
        then $\boldsymbol{F}_G = R^{ij}_G \boldsymbol{K}_G$
        by \cref{lem:identities}.
        Write $\alpha_C$ for the coefficients of $\boldsymbol{K}$.
        Then
        \[
            \boldsymbol{F}_G = R^{ij}_G \boldsymbol{K}_G
            = \sum_C \alpha_C R^{ij}_G \boldsymbol{1}_{C \cap P(G)^k}.
        \]
        By stability of the coloring, the vector $R^{ij}_G \boldsymbol{1}_{C \cap P(G)^k}$ is constant on color classes
        and hence a linear combination of color class indicator vectors,
        as desired.
        For $\boldsymbol{F} = \boldsymbol{K}^{\eta(i,j)}$, the claim is analogous.
    \end{claimproof}

    It remains to show the converse direction.
    \begin{claim}
        For every $\dwl_{G,k}^
        \infty$- or $\dwl_{H,k}^
        \infty$-color class $C$,
        there exist coefficients $\alpha_{\boldsymbol{F}}$
        for $\boldsymbol{F} \in \mathfrak{A}_k$
        such that 
        \[
            \boldsymbol{1}_{C \cap P(G)^k} = \sum_{\boldsymbol{F}} \alpha_{\boldsymbol{F}} \boldsymbol{F}_G, \quad\quad
            \boldsymbol{1}_{C \cap P(H)^k} = \sum_{\boldsymbol{F}} \alpha_{\boldsymbol{F}} \boldsymbol{F}_H.
        \]
    \end{claim}
    \begin{claimproof}
        The proof is by induction on the number of $\dwl$-iterations.
        For the initial coloring~\eqref{eq:initial}, which depends only on the tuple of sizes,
        let $\boldsymbol{F}^{\ell_1, \dots, \ell_k} \in \mathfrak{A}_k$
        denote the edge-less $k$-partitioned graph with $\ell_i$ vertices of color~$i$ for $i \in [k]$.
        For $\boldsymbol{D} \in P(G)^k$, it holds that
        \[
            \boldsymbol{F}^{\ell_1, \dots, \ell_k}_G(\boldsymbol{D}) = |D_1|^{\ell_1} \cdots |D_k|^{\ell_k}.
        \]
        This is a multivariate polynomial in variables $|D_1|, \dots, |D_k|$.
        Hence, there exist coefficients $\alpha_{\ell_1, \dots, \ell_k}$ such that $\sum \alpha_{\ell_1, \dots, \ell_k}\boldsymbol{F}^{\ell_1, \dots, \ell_k}_G(\boldsymbol{D})$
        is the indicator vector on a given $\dwl$-initial color.
        
        For the update step~\eqref{eq:refine},
        first observe that $R^{ij}_G \boldsymbol{1}_C$ is a linear combination of homomorphism vectors.
        Indeed, we have by \cref{lem:identities} that
        \[
            R^{ij}_G \boldsymbol{1}_{C \cap P(G)^k} = \sum_{\boldsymbol{F}} \alpha_{\boldsymbol{F}} R^{ij}_G\boldsymbol{F}_G
            = \sum_{\boldsymbol{F}} \alpha_{\boldsymbol{F}} \boldsymbol{F}^{\rho(i \to j)}_G
        \]
        and analogously for $H$.
        A similar statement holds for $E^{ij}$.

        It remains to implement the coloring step: If two tuples $\boldsymbol{D}, \boldsymbol{D'} \in P(G)^k$ differ in an entry of e.g.\ $R^{ij}_G \boldsymbol{1}_C$, then they receive a distinct color. 
        To that end, we use polynomial interpolation as above. 
        Since the vectors $R^{ij}_G \boldsymbol{1}_C$ and $R^{ij}_H \boldsymbol{1}_C$ contain only finitely many values, there exists, for every $a \in \mathbb{Q}$, a univariate polynomial $p \in \mathbb{Q}[x]$ such that $p(R^{ij}_G \boldsymbol{1}_C)$ and $p(R^{ij}_H \boldsymbol{1}_C)$ when evaluated entrywise yield the indicator vectors on those values that are equal to~$a$.
        These polynomial evaluations can again be written as linear combinations of homomorphism vectors for graphs in $\mathfrak{A}_k$ by~\eqref{eq:disjoint-union}.
        By taking entrywise products of these indicator vectors and applying~\eqref{eq:disjoint-union},
        we may construct indicator vectors of $\dwl$-colors.
    \end{claimproof}

    The claim implies that if $G$ and $H$ are homomorphism indistinguishable over all graphs of cliquewidth~$\leq k$,
    then they agree in the color of $(V(G), \dots, V(G))$ and $(V(H), \dots, V(H))$ by \cref{obs:drop-color}.
\end{proof}

As a corollary, we observe that the dense Weisfeiler--Leman algorithm is at least as powerful as the original Weisfeiler--Leman algorithm (\cref{cor:wl-dwl}) and can be arbitrarily more powerful on specific graphs (\cref{cor:wl-dwl-2-to-infty}).

\begin{corollary}\label{cor:wl-dwl}
    Let $k \ge1 $. If graphs $G$ and $H$ are distinguished by the $k$-dimensional Weisfeiler--Leman algorithm, then they are distinguished by the $(2^{k+1}+2)$-dimensional dense Weisfeiler--Leman algorithm.
\end{corollary}
\begin{proof}
    By \cite[Corollary~5.8]{courcelle_upper_2000},
    every graph of treewidth~$\leq k$ has cliquewidth~$\leq 2^{k+1}+2$.
    Now the corollary follows from \cref{thm:dwl-hom} and \cite{dvorak_2010_homomorphisms}.
\end{proof}

\begin{corollary}\label{cor:wl-dwl-2-to-infty}
    For every $k \geq 1$,
    there exist graphs $G$ and $H$ that are not distinguished by the $k$-dimensional Weisfeiler--Leman algorithm but distinguished by the $2$-dimensional dense Weisfeiler--Leman algorithm.
\end{corollary}
\begin{proof}
    Take $G$ and $H$ to be the \textsmaller{CFI} graphs \cite{cai_furer_immerman_1992} of the $(k+2)$-vertex complete graph $K_{k+2}$ as constructed in \cite{roberson_oddomorphisms_2022}.
    By~\cite{neuen_cfi_hom_counts}, 
    $G$ and $H$ are not distinguished by the $k$-dimensional WL algorithm as $K_{k+2}$ has treewidth~$k+1$.
    By~\cite[Corollary~3.7]{roberson_oddomorphisms_2022}, they are distinguished by homomorphism counts from $K_{k+2}$, which is of cliquewidth~$2$.
    Hence, they are distinguished by the $2$-dimensional dense WL algorithm by \cref{thm:dwl-hom}. 
\end{proof}

\section{Graph classes of bounded cliquewidth}
\label{sec:automata}
Having shown that homomorphism indistinguishability over the class of all graphs of cliquewidth~$\leq k$ is in $\mathsf{EXPTIME}$,
we now turn to other graph classes of bounded cliquewidth.
As discussed in the introduction, if homomorphism indistinguishability over a graph class~$\mathcal{F}$ is decidable, the same does not necessarily hold for its sub- or superclasses.

We show that homomorphism indistinguishability over every $\mathsf{CMSO}_1$-definable graph class of bounded cliquewidth is decidable.
Our results are inspired by previous work of
\textcite{cerny_homomorphism_2026}, who showed that homomorphism indistinguishability over $\mathsf{CMSO}_2$-definable graph classes of bounded pathwidth and treewidth reduces to equivalence testing for multiplicity word and tree automata, respectively.
We extend this observation to the dense case.
Here, pathwidth and treewidth are replaced by linear cliquewidth and cliquewidth;
$\mathsf{CMSO}_1$ plays the role of $\mathsf{CMSO}_2$.

We reduce testing homomorphism indistinguishability over $\mathsf{CMSO}_1$-definable graph class of bounded cliquewidth
to deciding equivalence of succinctly given multiplicity tree automata (\MTAs) with exponentially many states. 
Given the input graphs $G$ and $H$,
the \MTAs are constructed as a product of an \MTA $\mathcal{A}_{G, k}$ that counts homomorphisms from cliquewidth~$\leq k$ graphs into $G$ and an \MTA $\mathcal{A}_{\phi, k}$ that depends only on the $\mathsf{CMSO}_1$-sentence~$\phi$ defining the graph class.
The latter automaton, whose existence follows from Courcelle's theorem \cite{courcelle_graph_2012}, is used to filter out the homomorphism counts from graphs satisfying~$\phi$.

\subsection{Bounded cliquewidth}
Let $T_k$ denote the ranked alphabet of cliquewidth-$\leq k$
expressions, i.e.\ containing the $0$-ary symbols $\bullet_i$ for $i \in [k]$, the $1$-ary symbols $\eta(i,j)$ and $\rho(i \to j)$ for $1 \leq i \neq j \leq k$,
and the $2$-ary symbol $\oplus$.
Trees over this alphabet encode cliquewidth-$k$ expressions.
We first observe that the linear representation of~$\mathfrak{A}_k$ constructed in \cref{lem:identities} yields the following automaton:

\begin{lemma}\label{lem:mta}
    For every graph $G$ and $k \geq 1$, 
    there is an \MTA $\mathcal{A}_{G, k}$ with state set $P(G)^k$ and alphabet $T_k$ such that for every tree $t$ over $T_k$ encoding a $k$-partitioned graph $\boldsymbol{F} = (F; C_1, \dots, C_k)$,
    \[
        \den{\mathcal{A}_{G, k}}(t) = \hom(F, G).
    \]
\end{lemma}
\begin{proof}The $0$-ary symbols $\bullet_i$ are represented by the vector $(X_1, \dots, X_k) \mapsto |X_i|$ in $\mathbb{Q}^{P(G)^k}$, see \cref{ex:bullet}.
    The $1$-ary symbols $\eta(i,j)$ and $\rho(i \to j)$ for $i \neq j$
    are represented by $E^{ij}_G, R^{ij}_G \in \mathbb{Q}^{P(G)^k \times P(G)^k}$ respectively, as defined in \cref{eq:Rij,eq:Eij}.
    For the $2$-ary symbol $\oplus$, we use the bilinear entrywise product map $\mathbb{Q}^{P(G)^k} \times \mathbb{Q}^{P(G)^k} \to \mathbb{Q}^{P(G)^k}$.
    The final vector $\gamma \in \mathbb{Q}^{P(G)^k}$ is the standard basis vector of the coordinate~$(V(G), \dots, V(G))$.
    It follows from \cref{lem:identities} that the semantics of $\mathcal{A}_{G, k} $ are as desired.
\end{proof}

We combine this automaton with Courcelle's automaton recognizing graphs that satisfy a given $\mathsf{CMSO}_1$-sentence.
Note that for properties of bounded-treewidth graphs $\mathsf{CMSO}_2$-definability and recognizability by tree automata is equivalent, as conjectured by \textcite{courcelle_monadic_1990} and proven by \textcite{bojanczyk_definability_2016}.
For properties of bounded-cliquewidth graphs, only the forward implication is known (\cref{thm:courcelle}); the backward direction is open \cite{bojanczyk_definable_2021}.
\begin{theorem}[Courcelle~\cite{courcelle_graph_2012}]\label{thm:courcelle}
    For every $\mathsf{CMSO}_1$-sentence $\phi$ and $k \geq 1$,
    there exists an \MTA $\mathcal{A}_{\phi, k}$ reading trees over $T_k$ such that, for every tree $t$ encoding some $k$-partitioned graph $\boldsymbol{F}$,
    \[
        \den{\mathcal{A}_{\phi, k}}(t) = \begin{cases}
            1, & \text{if } \boldsymbol{F} \text{ satisfies } \phi,\\
            0, & \text{otherwise}.
        \end{cases}
    \]
\end{theorem}

\cref{lem:mta,thm:courcelle} readily imply that testing $G$ and $H$ for homomorphism indistinguishability over the class $\mathcal{T}_{\phi, k}$ of graphs of cliquewidth~$\leq k$ satisfying~$\phi$ reduces to testing whether the product \MTAs  $\mathcal{A}_{\phi, k} \otimes \mathcal{A}_{G, k}$ and $\mathcal{A}_{\phi, k} \otimes \mathcal{A}_{H, k}$ are equivalent.
We show that the latter can be done in randomized exponential time.
To that end, we first bound the size of the numbers arising in the computation using a classical result of \textcite{seidl_deciding_1990}.
For background on product \MTAs, see \cite{marusic_complexity_2015,cerny_homomorphism_2026}.
\begin{lemma}\label{cor:mtas-clique-width-mso}
    Let $k \geq 1$ and let $\phi$ be a $\mathsf{CMSO}_1$-sentence.
    Write $C_{\phi, k}$ for the number of states in the automaton $\mathcal{A}_{\phi, k}$ from \cref{thm:courcelle}.
    For two $n$-vertex graphs $G$ and $H$, the following are equivalent:
    \begin{enumerate}
        \item $G$ and $H$ are homomorphism indistinguishable over $\mathcal{T}_{\phi, k}$,
        \item $G$ and $H$ are homomorphism indistinguishable over the graphs in $\mathcal{T}_{\phi, k}$ on at most $2^{C_{\phi, k} \cdot  2^{nk +1}}$ vertices,
        \item the \MTAs $\mathcal{A}_{\phi, k} \otimes \mathcal{A}_{G, k}$ and $\mathcal{A}_{\phi, k} \otimes \mathcal{A}_{H, k}$ are equivalent.
    \end{enumerate}
\end{lemma}

\begin{proof}The equivalence of the first and last item is immediate from \cref{lem:mta,thm:courcelle}.
    Indeed, for a tree $t$ over $T_k$ encoding some $k$-partitioned graph $\boldsymbol{F} = (F; C_1, \dots, C_k)$,
    \[
       \den{\mathcal{A}_{\phi, k} \otimes \mathcal{A}_{G, k}}(t)
       = \den{\mathcal{A}_{\phi, k}}(t) \cdot \den{\mathcal{A}_{G, k}}(t)
       = \begin{cases}
           \hom(F, G) , & \text{if } \boldsymbol{F} \text{ satisfies } \phi,\\
            0, & \text{otherwise},
       \end{cases}
    \]
    and analogously for $H$.
    Hence, the \MTAs $\mathcal{A}_{\phi, k} \otimes \mathcal{A}_{G, k}$ and $\mathcal{A}_{\phi, k} \otimes \mathcal{A}_{H, k}$ are equivalent if, and only if, $G$ and $H$ are homomorphism indistinguishable over $\mathcal{T}_{\phi, k}$.
    
    For the second item, by \cite[Theorem~4.2]{seidl_deciding_1990},
    the \MTAs $\mathcal{A}_{\phi, k} \otimes \mathcal{A}_{G, k}$ and $\mathcal{A}_{\phi, k} \otimes \mathcal{A}_{H, k}$  are equivalent if, and only if, for every tree $t$ over $T_k$ of depth $\leq 2 N$, it holds that
    \(
    \den{\mathcal{A}_{\phi, k} \otimes \mathcal{A}_{G, k}}(t)
    =
    \den{\mathcal{A}_{\phi, k} \otimes \mathcal{A}_{H, k}}(t).
    \)
    Here, $N \coloneqq C_{\phi, k} 2^{nk}$ is the number of states in each of the \MTAs.
    Moreover, we follow the convention of~\cite{seidl_deciding_1990} stipulating that the one-vertex tree has depth~$0$.
    Finally, observe that any such tree $t$ is comprised of vertices of out-degree $0$, $1$, or $2$.
    Hence, it has at most $2^{2N} = 2^{C_{\phi,k}2^{nk+1}}$ leaves, each corresponding to a vertex in the bounded-cliquewidth graph represented by~$t$.
\end{proof}
This concludes the preparations for the proof of \cref{thm:main-clique-width-mso}.
\mainCWMSO*
\begin{proof}
    We test homomorphism indistinguishability over $\mathcal{F}= \mathcal{T}_{\phi, k}$ for some hardcoded $\mathsf{CMSO}_1$-sentence $\phi$ and an integer~$k \geq 1$.
    By \cite{marusic_complexity_2015},
    there exists a randomized algorithm with one-sided error for testing \MTA equivalence that runs in polynomial time in the number of states.
    In our case, the number of states is $C_{\phi, k} 2^{nk}$ for $n$-vertex input graphs with $C_{\phi, k}$ as in \cref{cor:mtas-clique-width-mso}.
    This yields the desired algorithm.

    More explicitly, one can perform the basis saturation algorithm of \cite{seppelt2024algorithmic}.
    Randomization is necessary here because numbers of superexponential bit length may arise during the computation.
    To avoid this, 
    we perform computations modulo some fixed prime~$p$ on exponentially many bits in $n$.
    All arithmetic operations can then be carried out in exponential time.
    Hence, we may decide in exponential time whether $G$ and $H$ are homomorphism indistinguishable over $\mathcal{T}_{\phi, k}$
    with homomorphisms counted modulo~$p$.
    By a Chinese Remaindering argument \cite[Lemma~9.2.4]{seppelt_homomorphism_2024},
    when we sample a random prime $p$ in the range
    \(
        2^{C_{\phi, k} \cdot  2^{nk +1}} \log n< p \leq \left( 2^{C_{\phi, k} \cdot  2^{nk +1}} \log n\right)^2,
    \)
    then the probability that $G$ and $H$ are homomorphism indistinguishable over $\mathcal{T}_{\phi, k}$ modulo~$p$ even though $G \not\equiv_{\mathcal{T}_{\phi, k}} H$
    is negligible.
\end{proof}

\subsection{Bounded linear cliquewidth}
In the case of bounded linear cliquewidth,
we use \MWAs instead of \MTAs.
For $k\geq 1$, let $\Sigma_k$ denote the alphabet of linear cliquewidth-$\leq k$ expressions, i.e.\ the set containing  the symbols $\oplus {\bullet_i}$, $\eta(i,j)$, and $\rho(i \to j)$ for $1 \leq i \neq j \leq k$.
Words over $\Sigma_k$ are precisely linear cliquewidth-$\leq k$ expressions.
The \MTAs $\mathcal{A}_{G, k}$ and $\mathcal{A}_{\phi,k}$ over $T_k$
restrict to \MWAs $\mathcal{A}'_{G, k}$ and $\mathcal{A}'_{\phi,k}$ over $\Sigma_k$ with the analogous properties.
For a $\mathsf{CMSO}_1$-sentence~$\phi$, 
let $\mathcal{W}_{\phi,k}$ denote the class of graphs of linear cliquewidth~$\leq k$ that satisfy~$\phi$.
The following lemma is analogous to \cref{cor:mtas-clique-width-mso}.

\begin{lemma}\label{cor:mwas-linear-clique-width-mso}
    Let $k \geq 1$ and let $\phi$ be a $\mathsf{CMSO}_1$-sentence.
    Write $C_{\phi, k}$ for the number of states in the automaton $\mathcal{A}_{\phi, k}$ from \cref{thm:courcelle}.
    For two $n$-vertex graphs $G$ and $H$, the following are equivalent:
    \begin{enumerate}
        \item $G$ and $H$ are homomorphism indistinguishable over $\mathcal{W}_{\phi,k}$,
        \item $G$ and $H$ are homomorphism indistinguishable over all graphs in $\mathcal{W}_{\phi,k}$ on $\leq C_{\phi, k} \cdot  2^{nk +1}$ vertices,
        \item the \MWAs $\mathcal{A}'_{\phi, k} \otimes \mathcal{A}'_{G, k}$ and $\mathcal{A}'_{\phi, k} \otimes \mathcal{A}'_{H, k}$ are equivalent.
    \end{enumerate}
\end{lemma}
\begin{proof}Analogous to \cref{cor:mtas-clique-width-mso}.
    The bound in the second item is improved by observing that, when $\Sigma_k$ is viewed as a ranked alphabet, then all symbols are of arity~$\leq 1$. Hence, the trees of depth $\leq 2N$ encode bounded-linear-cliquewidth graphs on at most~$2N$ vertices.
\end{proof}

Equipped with \cref{cor:mwas-linear-clique-width-mso},
we design a $\mathsf{PSPACE}$-algorithm for testing equivalence of succinctly given \MWAs with exponentially many states.

\mainLCWMSO*
\begin{proof}
    Write $\mathcal{F} = \mathcal{W}_{\phi,k}$ for hardcoded~$\phi$ and~$k$.
    Write $n$ for the number of vertices in $G$ and $H$.
    Let $N \coloneqq C_{\phi, k} 2^{nk}$ be the number of states in the \MWAs $\mathcal{A}'_{\phi, k} \otimes \mathcal{A}'_{G, k}$ and $\mathcal{A}'_{\phi, k} \otimes \mathcal{A}'_{H, k}$.
    By \cref{cor:mwas-linear-clique-width-mso}, 
    we must test whether there exists a word
    $w \in \Sigma_k^{\leq 2N}$ such that 
    \(
    \den{\mathcal{A}'_{\phi, k} \otimes \mathcal{A}'_{G, k}}(w)  \neq \den{\mathcal{A}'_{\phi, k} \otimes \mathcal{A}'_{H, k}}(w).
    \)
    The value $\den{\mathcal{A}'_{\phi, k} \otimes \mathcal{A}'_{G, k}}(w)$ is a product of $\leq 2N$ many $N \times N$
    matrices, to which we have implicit access.
    Thus, the value $\den{\mathcal{A}'_{\phi, k} \otimes \mathcal{A}'_{G, k}}(w)$ 
    can be determined by splitting $w$ into two equal-length parts $w = w_1 w_2$
    and computing intermediate $N \times N$ matrices for $w_1$ and $w_2$.
    This yields a recursion tree of depth $\log(2N) = \poly(n)$.
    Intermediate values are computed when needed and then discarded using indices of length $O(\log N) = \poly(n)$.
    In this way, we never store more than $\poly(n)$ numbers.

    The only remaining issue is to deal with the bit~length of the involved numbers.
    The value $\den{\mathcal{A}'_{\phi, k} \otimes \mathcal{A}'_{G, k}}(w)$ can be as large as $\poly(N)^{2N}$ requiring $N \poly(n)$ bits.
    In order to overcome this problem,
    we compute $\den{\mathcal{A}'_{\phi, k} \otimes \mathcal{A}'_{G, k}}(w)$ modulo primes $p$ of bit length $\poly(n)$ and apply the Chinese Remainder Theorem \cite[Lemma~9.2.6]{seppelt_homomorphism_2024}.
    This yields a $\mathsf{PSPACE}$ algorithm.
\end{proof}

\section{\texorpdfstring{$\mathsf{PSPACE}$}{PSPACE}-hardness}
\label{sec:hardness}
In this section, we prove \cref{thm:pspace_hard}, which
constructs an explicit set $\mathcal F$ of $\mathsf{MSO}_1$-definable vertex-colored graphs of constant linear cliquewidth such that the problem of homomorphism indistinguishability over $\mathcal F$ is $\mathsf{PSPACE}$-hard. 
To obtain hardness, we reduce from the problem $\SlidingTokens$, shown to
be $\mathsf{PSPACE}$-hard by \textcite{hearn_demaine_2005}. 
The input to $\SlidingTokens$ is a graph
$G$, on whose vertices are placed $k$ tokens, including one distinguished token, so that no two tokens are on adjacent vertices. The problem is to determine whether there is a sequence of moves of tokens to adjacent unoccupied vertices,
all the while preserving the pairwise non-adjacency of the tokens, leading to a move of the distinguished token
off of its starting vertex.

To model $\SlidingTokens$ as a homomorphism indistinguishability problem, we first construct a fixed family  
$\mathcal F$ of vertex-colored pattern graphs, each
of which consists of a path-like sequence of blocks (see \cref{fig:pattern}). 
Each block models one configuration of the tokens on the
graph $G$ in a sequence of moves. Then, given an instance $I$ of $\SlidingTokens$, we construct a vertex-colored host graph $X_I$ (see \cref{fig:host}) so that $I$ is a
yes-instance (that is, admits a desired sequence of moves) if and only if there is an $F \in \mathcal F$ admitting
a nonzero number of homomorphisms to $X_I$. Roughly, $X_I$ consists of two layers designed to force any homomorphism
from an $F \in \mathcal F$ to map successive blocks of $F$ to alternating layers of $X_I$. Edges (really,
non-edges) are arranged within each layer of $X_I$ to enforce that no configuration has tokens on adjacent vertices,
and between the two layers of $X_I$ to enforce that every move sends a token to an adjacent unoccupied vertex.
Then $X_I$ and the trivial graph admitting no homomorphisms from any graph in $\mathcal F$ are homomorphism indistinguishable over~$\mathcal F$ if and only if $I$ is a no-instance of $\SlidingTokens$.

\subsection{Setup and graph construction}

We prove \cref{thm:pspace_hard} by reduction from the reconfiguration problem $\SlidingTokens$.
A \emph{valid configuration} of $k$ tokens in a graph $G=(V,E)$ is a tuple
$C=(c_1,\ldots,c_k)\in V^k$ whose entries are pairwise distinct and pairwise
non-adjacent.  An instance $I = (G,C)$ consists of $G$, a valid initial configuration
$C=(c_1,\ldots,c_k)$, and the distinguished token $1$.  In the original
formulation, a \emph{single-token move} replaces exactly one $c_j$ by a
neighbor $c'_j\in N_G(c_j)$ such that the resulting tuple is again a valid
configuration.  The question is whether there is a sequence of single-token
moves from $C$ to some valid configuration
$C'=(c'_1,\ldots,c'_k)$ with $c'_1\ne c_1$.  Equivalently, we may stop the
first time token $1$ moves and ask whether there is a sequence ending with
$c'_1\in N_G(c_1)$.

\begin{theorem}[\textcite{hearn_demaine_2005}]\label{thm:sliding-hard}
$\SlidingTokens$ is $\mathsf{PSPACE}$-complete.
\end{theorem}

For the reduction it is convenient to allow parallel moves.  
\begin{definition}\label{def:move}
A \emph{move}
from a valid configuration $C_i=(c_{i,1},\ldots,c_{i,k})$ to a valid configuration
$C_{i+1}=(c_{i+1,1},\ldots,c_{i+1,k})$ satisfies the following additional 
conditions:
\begin{enumerate}
    \item for each $j\in [k]$, either $c_{i,j}=c_{i+1,j}$ or $c_{i,j} c_{i+1,j}\in E(G)$, and
    \item for all $j,j'\in [k]$ with $j\ne j'$, we have $c_{i,j} \ne c_{i+1,j'}$ and
    $c_{i,j} c_{i+1,j'}\notin E(G)$.
\end{enumerate}
\end{definition}

\begin{lemma}\label{lem:move-equivalence}
The source instance is positive in the original single-token sense if and only
if there is a sequence of moves from the initial configuration
$(c_1,\ldots,c_k)$ to a valid configuration whose first coordinate lies in
$N_G(c_1)$.
\end{lemma}

\begin{proof}
Every single-token move is a move, so one implication follows by stopping an
ordinary solution sequence immediately after token $1$ first moves.  Conversely,
consider one move from $C_i$ to $C_{i+1}$.  Move the tokens that change position one at
a time, in any order.  At every intermediate stage, any pair of tokens is either
old--old, new--new, or old--new.  The old--old pairs are valid because $C_i$ is a
valid configuration, the new--new pairs are valid because $C_{i+1}$ is a valid
configuration, and the mixed pairs are valid by the second old--new condition.
Thus the move expands into at most $k$ single-token moves.  Expanding each move
in a sequence proves the converse implication.
\end{proof}

We construct a color set $\Gamma$ and a fixed pattern class $\mathcal F$ of $\Gamma$-colored graphs of linear cliquewidth $O(1)$ such that the following holds: For every instance of $\SlidingTokens$, we can efficiently construct two graphs that are distinguished by $\mathcal F$ iff the instance is positive.
The pattern class \(\mathcal F\) will consist of path-like pattern graphs whose
blocks represent successive token configurations.  The construction is spelled out in \cref{sec:patterns}; for now it suffices to say that each block has 
\begin{itemize}
    \item a \textit{center vertex} (used to linearly order the blocks, as center vertices form a path),
    \item a clique of \textit{token vertices}, and 
    \item a \textit{counter path} (used to order the token vertices).
\end{itemize}

In \cref{sec:host-graphs},
given an instance \(I=(G,C)\) of \(\SlidingTokens\), we construct two host graphs $J_\Gamma$ and $X_I$. The graph $J_\Gamma$ is designed to trivially not admit any homomorphisms from the graphs in $\mathcal F$. The graph
\(X_I\) depends on $I$ and admits a homomorphism from a graph in $\mathcal F$ iff $I$ is a positive instance: It contains \textit{two} layers of state vertices, and the blocks from pattern graphs map into the two layers in an alternating fashion.
Edges within layers encode independence constraints, while edges between layers check validity of moves. Additional vertices ensure the required boundary conditions.

\subsubsection{Pattern graphs}\label{sec:patterns}
Let $\Gamma$ consist of the colors
\[
  \underbrace{\col z_0,\col z_1}_{\text{center colors}},
  \quad 
  \underbrace{\col t_0,\col t_1}_{\text{token colors}},
  \quad 
  \underbrace{\col{beg},\col{fin}}_{\text{anchor colors}},
  \quad 
  \underbrace{\col p^-,\col p^+,
  \col p_0,\col p_1,\col p_2}_{\text{colors for counter paths}} .
\]
The pattern class $\mathcal F$ consists of $\Gamma$-colored graphs $F(\boldsymbol r)$ for tuples $\boldsymbol r = (r_0,\ldots,r_m) \in \mathbb N^*$. See \cref{fig:pattern}.
Each graph $F(\boldsymbol r)$ consists of $m+1$ \emph{blocks} that are connected in a path-like manner. The
$i$th block consists of
\begin{itemize}
    \item a \emph{center} vertex $z_i$ of color $\col z_{i\bmod 2}$, 
    \item a clique $T_i=\{t_{i,1},\ldots,t_{i,r_i}\}$ of \emph{token} vertices, all of color $\col t_{i\bmod 2}$, 
    \item an induced path $P_i= p^-_i,p_{i,1},\ldots, p_{i,r_i},p^+_i$
        called the \emph{counter path} of the $i$-th block.
        Vertex $p^-_i$ has color $\col p^-$ and $p^+_{i}$ has color $\col p^+$. 
        The remaining vertices $p_{i,j}$ have color $\col p_{j\bmod 3}$.
    \item edges between the center $z_i$ and every vertex of $T_i \cup P_i$. 
    \item for all $j\in\{1,\ldots,r_i\}$, an edge between $t_{i,j}$ and $p_{i,j}$.
\end{itemize}
Outside of the blocks, the following vertices and edges are present:
\begin{itemize}
    \item Each block $i<m$ is connected to the following block $i+1$ with the edge $z_i z_{i+1}$ and all edges between 
        $T_i$ and $T_{i+1}$.
    \item Two \emph{anchor} vertices $\alpha$ and $\omega$, of colors $\col{beg}$ and $\col{fin}$, respectively.  
    \item The vertex $\alpha$ is adjacent to $\{z_0\} \cup T_0$, and $\omega$ is adjacent to $\{z_m\} \cup T_m$.
\end{itemize}

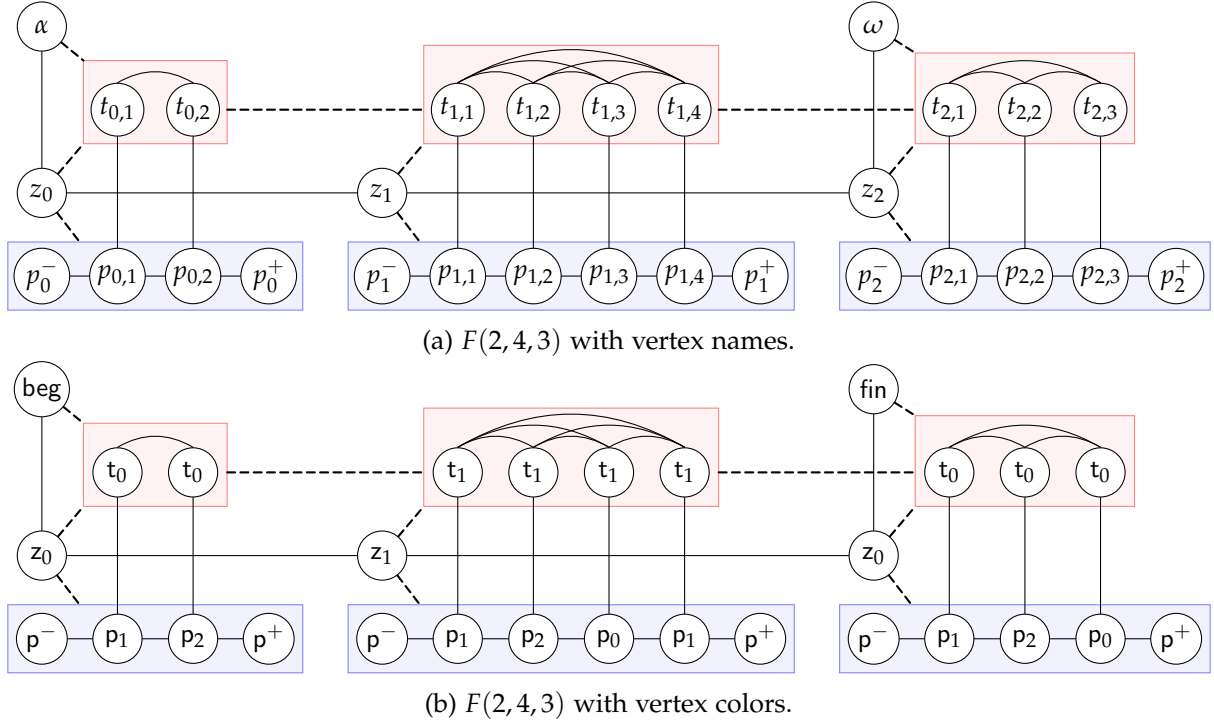
\begin{figure}
\centering
 \begin{subfigure}[b]{\textwidth}
    \centering
    \begin{tikzpicture}[
  x=1cm,y=1cm,
  vertex/.style={circle,draw,fill=white,inner sep=1.2pt,minimum size=6.5mm,font=\small},
  smallvertex/.style={circle,draw,inner sep=1pt,minimum size=6mm,font=\scriptsize},
  note/.style={font=\scriptsize,align=center},
  joinedge/.style={thick,densely dashed},
  thinedge/.style={gray},
  every path/.style={line cap=round}
]
    \def\ysh{2.2}
    \def\xsh{1}

    \def\rzero{2}
    \def\rone{4}
    \def\rtwo{3}
    \def\bgap{0.45}

    \foreach \ii/\rr/\shift in {0/\rzero/0,1/\rone/4.5,2/\rtwo/11} {
    \begin{scope}[xshift=\shift*\xsh cm]

        \draw[color=red!50,thin,fill=red!05] (\xsh-\bgap,\ysh-\bgap) rectangle (\rr*\xsh+\bgap,\ysh+\bgap+\rr/10);
        \draw[color=blue!50,thin,fill=blue!05] (-\bgap,-\bgap) rectangle (\rr*\xsh+\xsh+\bgap,\bgap);

        \node[vertex] (c\ii) at (0,\ysh/2) {$z_\ii$};

        \node[vertex] (pm\ii) at (0,0) {$p^-_{\ii}$};
        \node[vertex] (pp\ii) at (\rr*\xsh+\xsh,0) {$p^+_{\ii}$};
        \draw (pp\ii) -- (pm\ii);

        \foreach \xx in {1,2,...,\rr} {
            \node[vertex] (p\ii\xx) at (\xx*\xsh,0) {$p_{\ii,\xx}$};
            \node[vertex] (t\ii\xx) at (\xx*\xsh,\ysh) {$t_{\ii,\xx}$};

            \draw (p\ii\xx) -- (t\ii\xx);
        }

        \foreach \xx in {1,...,\rr} {
            \foreach \xxx in {\xx,...,\rr} {
                \draw (t\ii\xx.north) to [bend left] (t\ii\xxx.north);
            }
        }

        \draw[joinedge] (c\ii) -- (0.5*\xsh,\bgap);
        \draw[joinedge] (c\ii) -- (\xsh-\bgap,\ysh-\bgap);

    \end{scope}
    }

    \draw (c0) -- (c1);
    \draw (c1) -- (c2);

    \node[vertex] (alpha) at (0,1.5*\ysh) {$\alpha$};
    \node[vertex] (omega) at (11,1.5*\ysh) {$\omega$};
    \draw (alpha) -- (c0);
    \draw (omega) -- (c2);

    \draw[joinedge] (alpha) -- (\xsh-\bgap,\ysh+\bgap+\rzero/10);
    \draw[joinedge] (omega) -- (\xsh-\bgap+11,\ysh+\bgap+\rtwo/10);

    \draw[joinedge] ($(t02)+(\bgap,0)$) -- ($(t11)-(\bgap,0)$);
    \draw[joinedge] ($(t14)+(\bgap,0)$) -- ($(t21)-(\bgap,0)$);

\end{tikzpicture}
     \caption{$F(2,4,3)$ with vertex names.}
    \label{fig:pattern_image_labels}
    \end{subfigure}
    \hfill
 \begin{subfigure}[b]{\textwidth}
    \centering
    \begin{tikzpicture}[
  x=1cm,y=1cm,
  vertex/.style={circle,draw,fill=white,inner sep=1.2pt,minimum size=6.5mm,font=\small},
  smallvertex/.style={circle,draw,inner sep=1pt,minimum size=6mm,font=\scriptsize},
  note/.style={font=\scriptsize,align=center},
  joinedge/.style={thick,densely dashed},
  thinedge/.style={gray},
  every path/.style={line cap=round}
]
    \def\ysh{2.2}
    \def\xsh{1}

    \def\rzero{2}
    \def\rone{4}
    \def\rtwo{3}
    \def\bgap{0.45}

    \foreach \ii/\rr/\shift[evaluate=\ii as \im using {int(mod(\ii,2))}] in {0/\rzero/0,1/\rone/4.5,2/\rtwo/11} {
    \begin{scope}[xshift=\shift*\xsh cm]

        \draw[color=red!50,thin,fill=red!05] (\xsh-\bgap,\ysh-\bgap) rectangle (\rr*\xsh+\bgap,\ysh+\bgap+\rr/10);
        \draw[color=blue!50,thin,fill=blue!05] (-\bgap,-\bgap) rectangle (\rr*\xsh+\xsh+\bgap,\bgap);

        \node[vertex] (c\ii) at (0,\ysh/2) {$\col z_\im$};

        \node[vertex] (pm\ii) at (0,0) {$\col p^-$};
        \node[vertex] (pp\ii) at (\rr*\xsh+\xsh,0) {$\col p^+$};
        \draw (pp\ii) -- (pm\ii);

        \foreach \xx[evaluate=\xx as \mm using {int(mod(\xx,3))}] in {1,2,...,\rr} {
            \node[vertex] (p\ii\xx) at (\xx*\xsh,0) {$\col p_{\mm}$};
            \node[vertex] (t\ii\xx) at (\xx*\xsh,\ysh) {$\col t_{\im}$};
            \draw (p\ii\xx) -- (t\ii\xx);
        }

        \foreach \xx in {1,...,\rr} {
            \foreach \xxx in {\xx,...,\rr} {
                \draw (t\ii\xx.north) to [bend left] (t\ii\xxx.north);
            }
        }

        \draw[joinedge] (c\ii) -- (0.5*\xsh,\bgap);
        \draw[joinedge] (c\ii) -- (\xsh-\bgap,\ysh-\bgap);

    \end{scope}
    }

    \draw (c0) -- (c1);
    \draw (c1) -- (c2);

    \node[vertex] (alpha) at (0,1.5*\ysh) {$\col{beg}$};
    \node[vertex] (omega) at (11,1.5*\ysh) {$\col{fin}$};
    \draw (alpha) -- (c0);
    \draw (omega) -- (c2);

    \draw[joinedge] (alpha) -- (\xsh-\bgap,\ysh+\bgap+\rzero/10);
    \draw[joinedge] (omega) -- (\xsh-\bgap+11,\ysh+\bgap+\rtwo/10);

    \draw[joinedge] ($(t02)+(\bgap,0)$) -- ($(t11)-(\bgap,0)$);
    \draw[joinedge] ($(t14)+(\bgap,0)$) -- ($(t21)-(\bgap,0)$);

\end{tikzpicture}

     \caption{$F(2,4,3)$ with vertex colors.}
    \label{fig:pattern_image_colors}
    \end{subfigure}
\caption{Illustrating $F(\boldsymbol{r})$ for $\boldsymbol{r} = (2,4,3)$. The red (top) and blue (bottom) boxes contain the 
    token cliques $T_i$ and counter paths $P_i$, respectively.
    A dashed line between a vertex and box or between two boxes indicates a complete join (that is,
edges between every vertex on both sides). Here, \cref{fig:pattern_image_labels,fig:pattern_image_colors} show the same
graph, but with vertex name and color labels, respectively.}
\label{fig:pattern}
\end{figure}

\begin{lemma}\label{lem:F-lcw}
The class $\mathcal F$ has bounded linear
cliquewidth.
\end{lemma}

\begin{proof}
Build $F(\boldsymbol r)$ from left to right, keeping labels for the beginning anchor $\alpha$,
the previous center, the current center, the previous token clique, the current
token clique, the previous counter vertex, two temporary labels for freshly created
vertices, and old vertices to which no more incident edges will be added. 
Create the beginning anchor $\alpha$ and keep it under its own label until the first block is finished.

In block $i$, create the center $z_i$.  If $i=0$, join it to $\alpha$; 
otherwise join it to the previous center $z_{i-1}$. Then create the initial
counter vertex $p^-_i$, join it to $z_i$,
and keep it as the last counter
vertex.  For $j=1,\ldots,r_i$, create $t_{i,j}$ with a temporary token label.
Join it to $z_i$, to the current token clique, to the previous token clique if $i>0$, and, if $i=0$, to $\alpha$. Next create
$p_{i,j}$ with a temporary counter label, and join it to $z_i$, to $t_{i,j}$, and to the
last counter vertex. 
Retire the previous counter
vertex (that is, change its label to the label for old vertices), 
make $p_{i,j}$ the previous counter vertex, and put $t_{i,j}$ into
the current token-clique label. After the rest of the block has been completed, 
add the final counter vertex $p^+_i$, connect it to $z_i$ and to the previous counter vertex, and retire it.
Then retire the
old previous center and previous token clique, and make the current center and
current token clique the previous ones.  After the first block, also retire $\alpha$. 
Finally, after the last block, create the final anchor $\omega$
and join it to the previous center and the previous token clique.  This uses only
constantly many labels, independent of $\boldsymbol r$.
\end{proof}

\begin{lemma}\label{lem:F-definable}
The class $\mathcal F$ is $\mathsf{MSO}_1$-definable.
\end{lemma}

\begin{proof}
View colors as unary predicates. $\mathsf{MSO}_1$ can express that there is a unique $\col{beg}$-vertex
and a unique $\col{fin}$-vertex; that the center vertices induce a path whose
colors alternate $\col z_0,\col z_1$, starting at the center adjacent to the
beginning anchor $\alpha$ and ending at the center adjacent to the final anchor $\omega$; and that
every non-anchor, non-center vertex has a unique center neighbor.  This unique
center neighbor assigns the vertex to a block.

Moreover, $\mathsf{MSO}_1$ can express that the token vertices assigned to a center form a clique of the parity color matching
that center.  The counter vertices assigned to the center induce a path from a
unique $\col p^-$-vertex to a unique $\col p^+$-vertex, and the regular color condition on this colored path is 
$\mathsf{MSO}_1$-definable.

The initial and final counter vertices $p_i^-$ and $p_i^+$ have no token neighbors, every other counter vertex
has exactly one token neighbor in its block, and every token vertex has exactly
one counter neighbor in its block.  Consecutive blocks, as read from the center
path, have complete joins between their token cliques.  The two anchors have
exactly the prescribed neighbors, and no other edges occur.  These conditions
hold precisely for the graphs generated by the construction, hence they define
$\mathcal F$.
\end{proof}

\subsubsection{Host graphs}
\label{sec:host-graphs}

Let $I = (G,C)$ with valid initial configuration $C = (c_1,\ldots,c_k)$ and distinguished token $1$ be the input for
$\SlidingTokens$. We construct a $\Gamma$-colored graph $X_I$ as follows. Intuitively, $X_I$ consists of two layers, one colored
$\col t_0$ and the other $\col t_1$, to which the successive token cliques $T_i$ from the pattern graph $F(\boldsymbol{r})$ will be 
alternatingly mapped. A map from $T_i$ into a layer of $X_I$ represents the configuration after $i$ moves;
the layers in $X_I$ are connected to each other to enforce that each configuration is reachable from the previous configuration by a move in the sense of \cref{def:move}, and the layers in $X_I$ are connected internally to enforce that each configuration itself is valid.

The host graph $X_I$ also contains a copy of the counter path, set up to force corresponding vertices in the token cliques in a pattern graph
$F(\boldsymbol{r})$ to consistently represent the same token, and to ensure that only pattern graphs encoding $k$ tokens (that is,
$F(\boldsymbol{r})$ for $\boldsymbol{r} = (k,\ldots,k)$) admit homomorphisms to $X_I$. See \cref{fig:host}.
\begin{itemize}
    \item For every token index $j\in [k]$, every possible placement $v\in V(G)$ of that token, and every layer $\ell\in\{0,1\}$, create a     \emph{state} vertex $(j,v)^\ell$ of color $\col t_\ell$.
        \begin{itemize}
            \item Inside each layer $\ell\in\{0,1\}$, put an edge between $(j,u)^\ell$ and $(j',v)^\ell$ if and only if
            $j\ne j'$ and $u\ne v$ and $uv\notin E(G)$.  
            \item Between the two layers, put an edge between $(j,u)^0$ and $(j',v)^1$ if and only if (a) $j=j'$ and $u$ and $v$ are identical or
                adjacent in $G$, or (b) $j\ne j'$ and $u \neq v$ and $uv\notin E(G)$.
        \end{itemize}
    \item Add two \emph{center} vertices $z_0'$ and $z_1'$, of colors $\col z_0$ and
        $\col z_1$, with the edge $z_0' z_1'$. 
        \begin{itemize}
            \item For $\ell \in \{0,1\}$, add edges between $z'_\ell$ and every state vertex in layer $\ell$.
        \end{itemize}

    \item Add a \emph{counter path} $p^-,p'_1,\ldots,p'_k,p^+$ with endpoint colors $\col p^-$ and $\col p^+$, and with 
        each $p'_j$ colored $\col p_{j\bmod 3}$.  
    \begin{itemize}
        \item Add edges between all vertices in the counter path and the center vertices $z_0'$ and $z_1'$.
        \item For $1\le j\le k$, add an edge between $p'_j$ and $(j,v)^\ell$ for every $v \in V(G)$ and $\ell \in \{0,1\}$.
    \end{itemize}

    \item Add two \emph{anchors} $\alpha'$ and $\omega'$, of colors $\col{beg}$ and
    $\col{fin}$.  
    \begin{itemize}
        \item Add edge $\alpha'z'_0$. 
        \item Add edges $\omega'z'_0$ and $\omega'z'_1$.
        \item Add edges between $\alpha'$ and $(j,c_j)^0$ for $1\le j\le k$.  
        \item For both layers $\ell\in\{0,1\}$, add an edge between $\omega'$ and $(j,v)^\ell$ if and only if (a) $j \neq 1$ or
            (b) $j = 1$ and $v \in N(c_1)$.
    \end{itemize}
\end{itemize}
The final two points ensure that the tokens start according to the configuration $C$, and that token 1 moves off of its starting
vertex $c_1$ on the final move (recall \cref{lem:move-equivalence}), respectively.

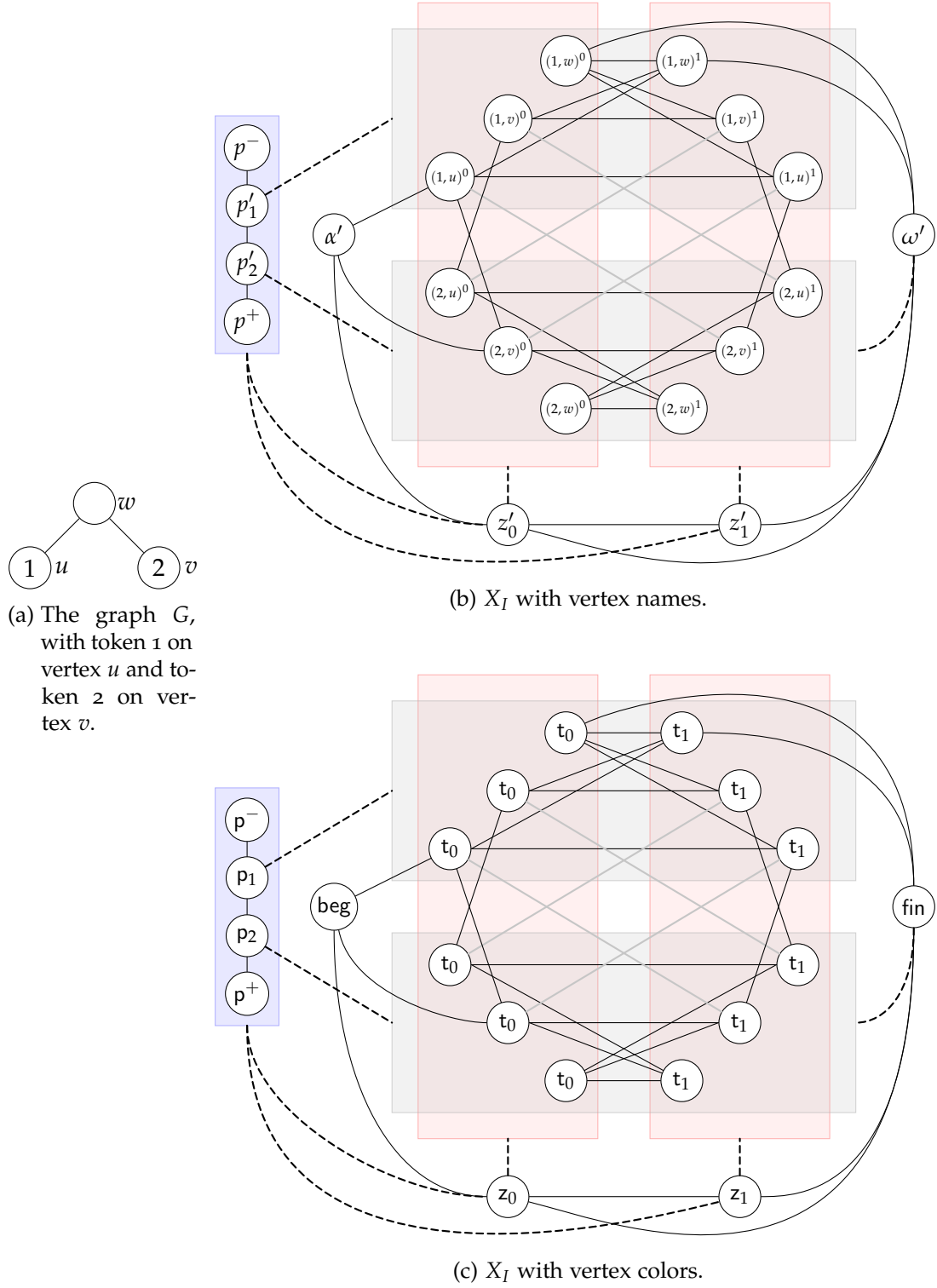
\begin{figure}
\centering
\begin{subfigure}{0.17\textwidth}
    \begin{tikzpicture}[
          vertex/.style={circle,draw,fill=white,inner sep=1.1pt,minimum size=6.5mm}
    ]
    \node[vertex] (w) at (0,0) {};
    \node[vertex] (u) at (-1,-1) {$1$};
    \node[vertex] (v) at (1,-1) {$2$};
    \draw (w) -- (v);
    \draw (u) -- (w);
    \node at ($(u)+(0.5,0)$) {$u$};
    \node at ($(v)+(0.5,0)$) {$v$};
    \node at ($(w)+(0.5,0)$) {$w$};
    \end{tikzpicture}
    \caption{The graph $G$, with token 1 on vertex $u$ and token 2 on vertex~$v$.}
    \label{fig:host_image_g}
    \vspace{8.5cm}
\end{subfigure}
\begin{subfigure}{0.7\textwidth}
    \begin{subfigure}{\textwidth}
        \centering
        \begin{tikzpicture}[
    scale=0.9,
  vertex/.style={circle,draw,fill=white,inner sep=1.1pt,minimum size=3.5mm,font=\tiny},
  lvertex/.style={circle,draw,fill=white,inner sep=1.2pt,minimum size=6.5mm,font=\small},
  down/.style={black},
  across/.style={thick,gray!45},
  side/.style={black},
  joinedge/.style={thick,densely dashed},
  every path/.style={line cap=round}
]
\def\bgap{0.55}
\def\bbgap{1}

\begin{scope}[rotate=90]
\draw[color=red!50,thin,fill=red!10,opacity=0.6] (-\bbgap,1-\bgap) rectangle (6+\bbgap,3+\bgap);
\draw[color=red!50,thin,fill=red!10,opacity=0.6] (-\bbgap,-1+\bgap) rectangle (6+\bbgap,-3-\bgap);

\draw[color=gray!60,thin,fill=gray!50,draw opacity=0.6, fill opacity = 0.2] (-\bgap,-3-\bbgap) rectangle (2+\bgap,3+\bbgap);
\draw[color=gray!60,thin,fill=gray!50,draw opacity=0.6, fill opacity = 0.2] (4-\bgap,-3-\bbgap) rectangle (6+\bgap,3+\bbgap);

\draw[color=blue!50,thin,fill=blue!15,opacity=0.6] (4.5+\bgap,6.5+\bgap) rectangle (1.5-\bgap,6.5-\bgap) ;

\node[vertex] (a0) at (0,1) {$(2,w)^0$};
\node[vertex] (b0) at (1,2) {$(2,v)^0$};
\node[vertex] (c0) at (2,3) {$(2,u)^0$};
\node[vertex] (cc0) at (4,3) {$(1,u)^0$};
\node[vertex] (bb0) at (5,2) {$(1,v)^0$};
\node[vertex] (aa0) at (6,1) {$(1,w)^0$};

\node[vertex] (a1) at (0,-1) {$(2,w)^1$};
\node[vertex] (b1) at (1,-2) {$(2,v)^1$};
\node[vertex] (c1) at (2,-3) {$(2,u)^1$};
\node[vertex] (cc1) at (4,-3) {$(1,u)^1$};
\node[vertex] (bb1) at (5,-2) {$(1,v)^1$};
\node[vertex] (aa1) at (6,-1) {$(1,w)^1$};

\draw[down] (a0) -- (a1);
\draw[down] (b0) -- (b1);
\draw[down] (c0) -- (c1);

\draw[down] (aa0) -- (aa1);
\draw[down] (bb0) -- (bb1);
\draw[down] (cc0) -- (cc1);

\draw[down] (b0.east) -- (a1);
\draw[down] (c0.east) -- (a1);
\draw[down] (b1.west) -- (a0);
\draw[down] (c1.west) -- (a0);

\draw[down] (bb0.east) -- (aa1);
\draw[down] (cc0.east) -- (aa1);
\draw[down] (bb1.west) -- (aa0);
\draw[down] (cc1.west) -- (aa0);

\draw[across] (b0) -- (cc1);
\draw[across] (c0) -- (bb1);
\draw[across] (c1) -- (bb0);
\draw[across] (b1) -- (cc0);

\draw[side] (b0) -- (cc0);
\draw[side] (c0) -- (bb0);
\draw[side] (b1) -- (cc1);
\draw[side] (c1) -- (bb1);

\node[lvertex] (z0) at (-2,2) {$z_0'$};
\node[lvertex] (z1) at (-2,-2) {$z_1'$};
\draw (z0) -- (z1);
\draw[joinedge] (z0) -- (-\bbgap,2);
\draw[joinedge] (z1) -- (-\bbgap,-2);

\node[lvertex] (alpha) at (3,5) {$\alpha'$};
\node[lvertex] (omega) at (3,-5) {$\omega'$};

\node[lvertex] (pm) at (1.5,6.5) {$p^+$};
\node[lvertex] (pp) at (4.5,6.5) {$p^-$};
\draw (pm) -- (pp);
\node[lvertex] (p1) at (2.5,6.5) {$p'_2$};
\node[lvertex] (p2) at (3.5,6.5) {$p'_1$};

\draw[joinedge] (p1) -- (1,3+\bbgap);
\draw[joinedge] (p2) -- (5,3+\bbgap);

\draw[joinedge] ($(pm)-(\bgap,0)$) .. controls +(-1.5,0) and +(0,2) .. (z0);
\draw[joinedge] ($(pm)-(\bgap,0)$) .. controls +(-5,0) and +(-0.5,2) .. (z1);

\draw (alpha) .. controls +(-3,0) and +(0,2) .. (z0);
\draw (alpha) .. controls +(-1.5,-0.3) and +(0,1) .. (b0);
\draw (alpha) -- (cc0);

\draw[joinedge] (omega) .. controls +(-1.5,0) and +(0,-0.5) .. (1,-3-\bbgap);

\draw (omega) .. controls +(-7.5,0) and +(-0.5,-2) .. (z0);
\draw (omega) .. controls +(-5,0) and +(0,-1) .. (z1);

\draw (omega) .. controls +(5,0) and +(0.5,-1) .. (aa0);
\draw (omega) .. controls +(3,0) and +(0,-1) .. (aa1);
\end{scope}
\end{tikzpicture}
         \vspace{-1.5cm}
        \caption{$X_I$ with vertex names.}
\label{fig:host_image_labels}
    \end{subfigure}
    \begin{subfigure}{\textwidth}
        \centering
        \begin{tikzpicture}[
    scale=0.9,
  vertex/.style={circle,draw,fill=white,inner sep=1.2pt,minimum size=6.5mm,font=\small},
  lvertex/.style={circle,draw,fill=white,inner sep=1.2pt,minimum size=6.5mm,font=\small},
  down/.style={black},
  across/.style={thick,gray!45},
  side/.style={black},
  joinedge/.style={thick,densely dashed},
  every path/.style={line cap=round}
]
\def\bgap{0.55}
\def\bbgap{1}

\begin{scope}[rotate=90]
\draw[color=red!50,thin,fill=red!10,opacity=0.6] (-\bbgap,1-\bgap) rectangle (6+\bbgap,3+\bgap);
\draw[color=red!50,thin,fill=red!10,opacity=0.6] (-\bbgap,-1+\bgap) rectangle (6+\bbgap,-3-\bgap);

\draw[color=gray!60,thin,fill=gray!50,draw opacity=0.6, fill opacity = 0.2] (-\bgap,-3-\bbgap) rectangle (2+\bgap,3+\bbgap);
\draw[color=gray!60,thin,fill=gray!50,draw opacity=0.6, fill opacity = 0.2] (4-\bgap,-3-\bbgap) rectangle (6+\bgap,3+\bbgap);

\draw[color=blue!50,thin,fill=blue!15,opacity=0.6] (4.5+\bgap,6.5+\bgap) rectangle (1.5-\bgap,6.5-\bgap) ;

\node[vertex] (a0) at (0,1) {$\col t_0$};
\node[vertex] (b0) at (1,2) {$\col t_0$};
\node[vertex] (c0) at (2,3) {$\col t_0$};
\node[vertex] (cc0) at (4,3) {$\col t_0$};
\node[vertex] (bb0) at (5,2) {$\col t_0$};
\node[vertex] (aa0) at (6,1) {$\col t_0$};

\node[vertex] (a1) at (0,-1) {$\col t_1$};
\node[vertex] (b1) at (1,-2) {$\col t_1$};
\node[vertex] (c1) at (2,-3) {$\col t_1$};
\node[vertex] (cc1) at (4,-3) {$\col t_1$};
\node[vertex] (bb1) at (5,-2) {$\col t_1$};
\node[vertex] (aa1) at (6,-1) {$\col t_1$};

\draw[down] (a0) -- (a1);
\draw[down] (b0) -- (b1);
\draw[down] (c0) -- (c1);

\draw[down] (aa0) -- (aa1);
\draw[down] (bb0) -- (bb1);
\draw[down] (cc0) -- (cc1);

\draw[down] (b0.east) -- (a1);
\draw[down] (c0.east) -- (a1);
\draw[down] (b1.west) -- (a0);
\draw[down] (c1.west) -- (a0);

\draw[down] (bb0.east) -- (aa1);
\draw[down] (cc0.east) -- (aa1);
\draw[down] (bb1.west) -- (aa0);
\draw[down] (cc1.west) -- (aa0);

\draw[across] (b0) -- (cc1);
\draw[across] (c0) -- (bb1);
\draw[across] (c1) -- (bb0);
\draw[across] (b1) -- (cc0);

\draw[side] (b0) -- (cc0);
\draw[side] (c0) -- (bb0);
\draw[side] (b1) -- (cc1);
\draw[side] (c1) -- (bb1);

\node[lvertex] (z0) at (-2,2) {$\col z_0$};
\node[lvertex] (z1) at (-2,-2) {$\col z_1$};
\draw (z0) -- (z1);
\draw[joinedge] (z0) -- (-\bbgap,2);
\draw[joinedge] (z1) -- (-\bbgap,-2);

\node[lvertex] (alpha) at (3,5) {$\col{beg}$};
\node[lvertex] (omega) at (3,-5) {$\col{fin}$};

\node[lvertex] (pm) at (1.5,6.5) {$\col p^+$};
\node[lvertex] (pp) at (4.5,6.5) {$\col p^-$};
\draw (pm) -- (pp);
\node[lvertex] (p1) at (2.5,6.5) {$\col p_2$};
\node[lvertex] (p2) at (3.5,6.5) {$\col p_1$};

\draw[joinedge] (p1) -- (1,3+\bbgap);
\draw[joinedge] (p2) -- (5,3+\bbgap);

\draw[joinedge] ($(pm)-(\bgap,0)$) .. controls +(-1.5,0) and +(0,2) .. (z0);
\draw[joinedge] ($(pm)-(\bgap,0)$) .. controls +(-5,0) and +(-0.5,2) .. (z1);

\draw (alpha) .. controls +(-3,0) and +(0,2) .. (z0);
\draw (alpha) .. controls +(-1.5,-0.3) and +(0,1) .. (b0);
\draw (alpha) -- (cc0);

\draw[joinedge] (omega) .. controls +(-1.5,0) and +(0,-0.5) .. (1,-3-\bbgap);

\draw (omega) .. controls +(-7.5,0) and +(-0.5,-2) .. (z0);
\draw (omega) .. controls +(-5,0) and +(0,-1) .. (z1);

\draw (omega) .. controls +(5,0) and +(0.5,-1) .. (aa0);
\draw (omega) .. controls +(3,0) and +(0,-1) .. (aa1);
\end{scope}
\end{tikzpicture}         \vspace{-1.5cm}
        \caption{$X_I$ with vertex colors.}
        \label{fig:host_image_colors}
    \end{subfigure}
\end{subfigure}
\caption{A host graph $X_I$ for $I = (G,C)$ where $G$ is the graph in 
    \cref{fig:host_image_g} and $C = (u,v)$. The red (vertical) boxes indicate the two layers, and
    the gray (horizontal) boxes indicate the vertices corresponding to tokens 1 and 2. 
    Again, a dashed edge indicates that a vertex is adjacent to every vertex in a box. The lighter edges
    connect vertices representing different tokens in different layers.}
\label{fig:host}
\end{figure}

Let $J_\Gamma$ be the edgeless $\Gamma$-colored graph with exactly one
vertex of each color.  The reduction outputs $(X_I,J_\Gamma)$.

\subsection{Correctness of the reduction}

The counter path in host $X_I$ forces every relevant counter path in the pattern to have length $k$.
\begin{lemma}\label{lem:counter}
A (color-preserving) homomorphism from a pattern counter path $P_i = p_i^-,p_{i,1},\ldots,p_{i,r},p_i^+$ with $r$ internal vertices to
$X_I$ exists if and only if $r=k$, where $k$ is the number of tokens in instance $I$.
In that case the image is exactly the host counter path $p^-,p_1',\ldots,p_k',p^+$.
\end{lemma}
\begin{proof}
If $r = k$, then, since the pattern and host counter paths have identical colors by construction,
the pattern path can be mapped bijectively onto the host path. This gives sufficiency; we now show necessity. 
Since the counter path contains all of the vertices in the host with the colors $\col p^-, \col p^+, \col p_0, \col p_1,
\col p_2$, the pattern counter path must map
entirely into the host counter path.
In the host $X_I$, only $p^-$ and $p^+$ are colored $\col p^-$ and $\col p^+$, 
so $p_i^-$ and $p_i^+$ must map to $p^-$ and $p^+$, respectively. 
If $r < k$, then the pattern path is shorter than the host path, which is impossible because, as just
stated, the endpoints of the pattern path must map to the endpoints of the host path. So assume $r \geq k$.
In $X_I$, the only vertex adjacent to $p^-$ is $p'_1$, so $p_{i,1}$ must map to $p'_1$. Now assume inductively that
$p_{i,j}$ has mapped to $p'_j$, with $j < \min\{r,k\}$. The next vertex $p_{i,j+1}$, with color $\col p_{j+1 \bmod{3}}$, must map 
either forward to $p'_{j+1}$, or backward to $p'_{j-1}$. But $p'_{j-1}$ has color $\col p_{j-1 \bmod{3}} \neq \col p_{j+1 \bmod{3}}$, 
so $p_{i,j+1}$ must map forward to $p'_{j+1}$. Therefore the first $k+1$ vertices (including $p_i^-$) of $P_i$ must map bijectively
onto $p^-,p'_1,\ldots,p'_k$. If $r > k$, then, similarly, vertex $p_{i,k+1} \in P_i$ must map onto $p'_{k-1}$ or $p^+$, both of
which have the wrong color. Thus $r = k$, and we are done.
\end{proof}

Let $F\in\mathcal F$, with blocks $0,\ldots,m$.  
Write $r_i$ for the length of the $i$-th counter path -- equivalently the size of the token clique $T_i$ -- in $F$. A
\emph{winning sequence of moves} for a $\SlidingTokens$ instance $I = (G,C)$ is a sequence of valid configurations
\[
  C_0=(c_{0,1},\ldots,c_{0,k}),\ldots,
  C_m=(c_{m,1},\ldots,c_{m,k})
\]
such that
\[
  C_0=C,\qquad c_{m,1}\in N_G(c_{0,1}),
\]
and $C_i\to C_{i+1}$ is a move in the sense of \cref{def:move} for every $i<m$.

\begin{lemma}\label{lem:hom-move-sequences}
Let $I = (G,C)$ be an instance of $\SlidingTokens$ with $k$ tokens.
For every $F(\boldsymbol r) \in\mathcal F$ with $m+1$ blocks, $\hom(F(\boldsymbol r),X_I)=0$ unless $r_i=k$ for every block $i$. In this case,
\[
  \hom(F(\boldsymbol r),X_I)
  =
  \#\{\text{length-$(m+1)$ winning sequences of moves for $I$}\}.
\]
In particular, $\hom(F,X_I)>0$ for some $F \in \mathcal F$ if and only if $I$ admits a winning sequence of moves.
\end{lemma}
\begin{proof}
By \cref{lem:counter}, each pattern counter path $P_i$ maps bijectively onto the host counter path, and each
$r_i = k$. In addition to $P_i$, the images of
the anchors $\alpha$ and $\omega$ and the center vertices $z_i$ are uniquely determined to be $\alpha'$,
$\omega'$, and $z'_{i \bmod 2}$, respectively, as these are the only vertices in the host with their
respective colors. 

It remains to reason about the images of the token cliques $T_i$. Let $F \in \mathcal F$ have $m+1$ blocks of size $k$,
and let $h:F\to X_I$ be a homomorphism. The token clique $T_i$ has color $\col t_{i\bmod 2}$, so it maps
into layer $i\pmod 2$. Every $t_{i,j} \in T_i$ is matched to a counter vertex $p_{i,j}$, which is mapped
to $p'_j$. The only state vertices of layer $i\pmod 2$ adjacent to $p'_j$
are those vertices $(j,v)^{i\bmod 2}$ associated with token $j$, for $v\in V(G)$. Therefore
\[
  h(t_{i,j})=(j,c_{i,j})^{i\bmod 2}
\]
for a uniquely determined $c_{i,j}\in V(G)$. This indicates that, in the $i$th configuration, token $j$
is on vertex $c_{i,j}$ of $G$. For $0 \leq i \leq m$, define
$C_i := (c_{i,1},\ldots,c_{i,k})$.
Since $T_i$ is a clique and same-layer edges of $X_I$ join $(j,c_{i,j})^{i \bmod 2}$ and 
$(j',c_{i,j'})^{i \bmod 2}$ if and only if $j \neq j'$ and $c_{i,j}$ and $c_{i,j'}$ are not adjacent in
$G$, the tuple $C_i$ is a valid token configuration. 

The complete join between $T_i$ and $T_{i+1}$ implies that every pair
\[ h(t_{i,j}),h(t_{i+1,j'}) = (j,c_{i,j})^{i \bmod 2}, (j',c_{i+1,j'})^{i+1 \bmod 2}\] is adjacent in $X_I$.
The two vertices of the pair are in different layers $i \pmod 2$ and $i+1 \pmod 2$.
If $j=j'$, the cross-layer definition says that $c_{i,j}=c_{i+1,j'}$ or $c_{i,j}c_{i+1,j'}\in E(G)$.  
If $j\ne j'$, the cross-layer definition says that $c_{i,j}$ and $c_{i+1,j'}$ are distinct and non-adjacent in 
$G$. These are exactly the two conditions for a move in \cref{def:move}.

Hence the same-layer edges enforce that every configuration $C_i$ is valid
and the cross-layer edges enforce that every step is a move. Finally,
the beginning anchor $\alpha$ is adjacent, among layer-$0$ state vertices of token index $j$, only
to $(j,c_j)^0$, where $C = (c_1,\ldots,c_k)$ is the initial configuration of $I$, so 
$C_0=C$, as desired. Similarly, the final anchor $\omega$ forces $c_{m,1}\in N_G(c_1) = N_G(c_{0,1})$ and
places no additional restriction on the other token positions. In this way, every homomorphism yields a
winning sequence of moves $C_0,\ldots,C_m$ for $I$.

Conversely, given a winning sequence of moves $C_0,\ldots,C_m$ for $I$, let 
$F := F(k,\ldots,k)$ have $m+1$ blocks. Construct a homomorphism from $F$ to $X_I$ as follows.
Map each center $z_i$ to $z'_{i \bmod 2}$, map every pattern counter path bijectively onto the host counter path,
map $\alpha$ and $\omega$ to $\alpha'$ and $\omega'$, and map
\[
      t_{i,j}\longmapsto (j,c_{i,j})^{i\bmod 2}.
\]
As above, the token clique edges inside one block are respected because each $C_i$ is a
valid configuration.  The joins between consecutive token cliques are respected
by the move conditions.  The matching edges $t_{i,j}p_{i,j}$ are respected
because $(j,c_{i,j})^{i\bmod 2}$ is adjacent to $p'_j$. The edges between the counter paths and center vertices
are present in both the pattern and host graphs. The edges incident to $\alpha$ and
$\omega$ are respected by the initial and accepting conditions of the
sequence of moves. Hence this is a homomorphism.

The two constructions are inverse to each other.  Therefore homomorphisms
$F\to X_I$ for $F$ with $m+1$ blocks are in bijection with 
length-$(m+1)$ winning sequences of moves for $I$.
\end{proof}

\begin{lemma}\label{lem:existence}
Instance $I$ is a \textsmaller{YES}-instance of $\SlidingTokens$ if and only if there
exists $F\in\mathcal F$ with $\hom(F,X_I)>0$.
\end{lemma}

\begin{proof}
If $I$ is a \textsmaller{YES}-instance, take the ordinary $\SlidingTokens$ sequence witnessing 
that $I$ is a \textsmaller{YES}-instance and stop it
immediately after the distinguished token first moves. This is a winning sequence
of moves, and the final configuration has first coordinate in $N_G(c_1)$.
Build a pattern $F\in\mathcal F$ with one length-$k$ block for each
configuration in the sequence.  \cref{lem:hom-move-sequences} gives
$\hom(F,X_I)>0$.

Conversely, if $\hom(F,X_I)>0$, then \cref{lem:hom-move-sequences}
gives a winning sequence of moves ending with the distinguished token in $N_G(c_1)$.
By \cref{lem:move-equivalence}, this expands into an ordinary $\SlidingTokens$ sequence in which the distinguished token moves.
\end{proof}

\begin{proof}[Proof of \cref{thm:pspace_hard}]
The construction of $X_I$ is polynomial in $|V(G)|+k$, and $J_\Gamma$ is
fixed.  Every graph in $\mathcal F$ has at least one edge, whereas
$J_\Gamma$ is independent.  Hence
\[
  \hom(F,J_\Gamma)=0
  \qquad\text{for every }F\in\mathcal F .
\]
By \cref{lem:existence}, the instance $I$ is a \textsmaller{YES}-instance of
$\SlidingTokens$ if and only if some $F\in\mathcal F$ has
$\hom(F,X_I)>0$.  Therefore
\[
  I\text{ is a \textsmaller{NO}-instance}
  \quad\Longleftrightarrow\quad
  X_I\equiv_{\mathcal F}J_\Gamma .
\]
This equivalence gives a polynomial reduction from the complement of $\SlidingTokens$
to the problem of homomorphism indistinguishability over $\mathcal F$, so, by \cref{thm:sliding-hard},
homomorphism indistinguishability over $\mathcal F$ is $\mathsf{PSPACE}$-hard. 
Finally,
$\mathcal F$ is $\mathsf{MSO}_1$-definable by \cref{lem:F-definable} and has constant linear cliquewidth
by \cref{lem:F-lcw}.
\end{proof}

\section{Hardness of homomorphism indistinguishability over cographs}

Our final lower bound concerns homomorphism indistinguishability over the class of graphs of cliquewidth $\leq 2$, which is also known as the class of \emph{cographs}.
\thmCographs*

We show \cref{thm:cographs} by fpt-reducing the colorful biclique detection problem to homomorphism indistinguishability over the class $\mathcal{B} \coloneqq \{K_{s,t} \mid s, t \geq 0\}$ of all bicliques.
The reduction uses inclusion--exclusion, which renders it fpt rather than polynomial-time.
Finally, we reduce homomorphism indistinguishability over bicliques to homomorphism indistinguishability over cographs. 
The crucial observation here is that bicliques are precisely the connected bipartite cographs.
This allows us to use a general reduction
\cite[Lemma~9.4.5]{seppelt_homomorphism_2024}
from homomorphism indistinguishability over the bipartite graphs in some graph class $\mathcal{F}$ to homomorphism indistinguishability over $\mathcal{F}$.

Curiously, our reduction does not rule out polynomial-time algorithms for testing homomorphism indistinguishability over the graphs of cliquewidth $\leq k$ for any $k \neq 2$.
Moreover, we suspect that homomorphism indistinguishability over bicliques is $\mathsf{C}_=\mathsf{P}$-complete as is the case for cliques \cite{boker_complexity_2019}.
Both problems remain open.

\begin{theorem}\label{thm:biclique-hard}
    Unless $\mathsf{FPT} = \mathsf{W[1]}$, there is no polynomial-time algorithm for deciding homomorphism indistinguishability over all bicliques.
\end{theorem}
\begin{proof}
    By \cite[Exercise~13.3]{cygan_parameterized_2015},
    the following \emph{colorful biclique detection problem} is $\mathsf{W[1]}$-hard:
    Given a graph $G$ with a partition $V(G) = A_1 \uplus \dots \uplus A_k \uplus B_1 \uplus \dots \uplus B_k$, decide whether there exist vertices $a_i \in A_i$ and $b_i \in B_i$ for $i \in [k]$ such that $a_i b_j \in E(G)$ for all $i,j \in [k]$.
    Here, the parameter is~$k$.
    We may assume that $G$ is bipartite with parts $A_1 \cup \dots \cup A_k$ and $B_1 \cup 
    \dots \cup B_k$.

    Observe that $G$ is a \textsmaller{YES}-instance to this problem if, and only if, there exists a biclique $K_{s,t}$ admitting a homomorphism $h \colon K_{s,t} \to G$ such that every set $A_1, \dots, A_k, B_1, \dots, B_k$ is hit by at least one vertex. 
    Here it is crucial that $G$ is bipartite with the specified bipartition.

    By inclusion--exclusion, comparing \cref{eq:surj-in-hom}, 
    the number of such homomorphisms is given by 
    \begin{equation}\label{eq:biclique-incl-excl}
        \sum_{L \subseteq [k]} (-1)^{k - |L|} \sum_{R \subseteq [k]} (-1)^{k - |R|} \hom(K_{s,t}, G[L, R])
    \end{equation}
    where, abusing notation, we write $G[L, R]$ for the subgraph of $G$ induced by $\bigcup_{\ell \in L}A_\ell \cup \bigcup_{r \in R} B_r$.

    We define two graphs $G^+$ and $G^-$ by collecting the positive and negative terms from \eqref{eq:biclique-incl-excl} respectively using disjoint unions.
    \[
        G^+ \coloneqq \sum_{\substack{L, R \subseteq [k] \\ |L| + |R| \equiv 0 \bmod 2}} G[L, R], \quad \quad\quad\quad
        G^- \coloneqq  \sum_{\substack{L, R \subseteq [k] \\ |L| + |R| \equiv 1 \bmod 2}} G[L, R].
    \]
    For connected graphs $F$ and graphs $H_1, H_2$, it holds that
    \begin{equation}
        \hom(F, H_1+ H_2) = \hom(F, H_1) + \hom(F, H_2).\label{eq:connected}
    \end{equation}
    Hence, the expression in \eqref{eq:biclique-incl-excl} is zero if, and only if, $\hom(K_{s,t}, G^+)= \hom(K_{s,t}, G^-)$.
    In particular, $G^+$ and $G^-$ are homomorphism indistinguishable over all bicliques if, and only if, $G$ is a \textsmaller{NO}-instance of the colorful biclique detection problem.
    The graphs $G^-$ and $G^+$ can be computed in fpt time. 
\end{proof}

Given \cref{thm:biclique-hard}, we describe a reduction from homomorphism indistinguishability over all bicliques to homomorphism indistinguishability over all cographs.
The proof requires the following combinatorial fact.
Note that we regard the one-vertex graph as a biclique and as connected.
\begin{fact}\label{lem:bipartite-cograph}
    Every bipartite connected cograph is a biclique.
\end{fact}
\begin{proof}
    Let $F$ be a connected bipartite cograph, with bipartition $(A,B)$,
and suppose that $F$ has at least two vertices.  By \cite[Theorem~2]{corneil_complement_1981}, $F$ has no induced $4$-vertex path~$P_4$.  If some $a\in A$ and $b\in B$ were
nonadjacent, a shortest $a$--$b$ path would have odd length at least
three.  Every shortest path is induced, so its first four vertices would
induce a $P_4$, a contradiction.  Thus every vertex of $A$ is adjacent
to every vertex of $B$.  Bipartiteness excludes edges inside either
part, and hence $F$ is complete bipartite.
\end{proof}

\begin{lemma}\label{lem:bicliques-to-cographs}
    Homomorphism indistinguishability over all bicliques polynomial-time many-one reduces to homomorphism indistinguishability over all cographs.
\end{lemma}
\begin{proof}
    By \cref{eq:connected},
    homomorphism indistinguishability over all cographs coincides with homomorphism indistinguishability over all connected cographs.
    By \cite[Lemma~9.4.5]{seppelt_homomorphism_2024},
    homomorphism indistinguishability over bipartite connected cographs polynomial-time many-one reduces to homomorphism indistinguishability over connected cographs.
    By \cref{lem:bipartite-cograph}, the former graph class is the class of bicliques.
\end{proof}

\Cref{thm:cographs} follows now by combining \cref{lem:bicliques-to-cographs,thm:biclique-hard}.

\section{A logic and a dense-to-sparse collapse theorem} \label{sec:logic_first10}

In this section, we design a logic capturing the power of the dense Weisfeiler--Leman algorithm.
Our logic should be understood as a dense analogue of first-order logic with counting quantifiers, which captures the power of the original Weisfeiler--Leman algorithm \cite{cai_furer_immerman_1992}.
We use our logic to prove a dense-to-sparse collapse theorem which compares the power of the dense and the original Weisfeiler--Leman algorithm.
Our \cref{thm:dense-to-sparse} extends a previous result of \textcite{dvorak_2010_homomorphisms}\footnote{We thank Zden\v{e}k Dvo\v{r}\'ak for pointing us to the \href{https://iti.mff.cuni.cz/series/2006/287.pdf}{preprint}, which discusses this argument in more detail.} asserting that graphs of girth $\ge 5$ that are not distinguished by the $(2k+1)$-dimensional Weisfeiler--Leman algorithm are homomorphism indistinguishable over the graphs of cliquewidth~$\leq k$.
\begin{theorem}\label{thm:dense-to-sparse}
    Let $k,t \geq 1$.
    If two $K_{t,t}$-subgraph-free graphs $G$ and $H$ are not distinguished by the $4kt$-dimensional Weisfeiler--Leman algorithm, 
    then they are not distinguished by the $k$-dimensional dense Weisfeiler--Leman algorithm.
\end{theorem}

Since the graphs constructed by \textcite{cai_furer_immerman_1992} to show that the $k$-dimensional Weisfeiler--Leman algorithm does not decide graph isomorphism are of degree $\leq 3$ and in particular $K_{4,4}$-subgraph-free, \cref{thm:dense-to-sparse} and \cite{cai_furer_immerman_1992} yield the following corollary.

\begin{corollary}\label{cor:not-iso}
    For every $k \geq 1$,
    there exist non-isomorphic graphs $G$ and $H$ on $O(k)$ vertices that are not distinguished by the $k$-dimensional dense Weisfeiler--Leman algorithm.
\end{corollary}

In light of \cref{thm:dwl-hom}, one may ask whether one can separate homomorphism indistinguishability over more dense graph classes more general than bounded cliquewidth from graph isomorphism.
We note that \textcite{dvorak_2010_homomorphisms} constructs a graph class of bounded twinwidth for which homomorphism indistinguishability coincides with isomorphism.
Hence, in contrast to cliquewidth, bounded twinwidth does not guarantee a separation from isomorphism, see \cite{neuen_distinguishing_2026}.

Towards the proof of \cref{thm:dense-to-sparse}, fix $k \geq 1$. 
We define a logic $\mathsf{CW}^k$ capturing homomorphism indistinguishability over the class of graphs of cliquewidth~$\leq k$.
We first define the syntax.
Each term carries a type $I \subseteq [k]$.

\begin{definition}[Syntax of $\mathsf{CW}^k$]\label{def:syntax}
        \begin{enumerate} 
                \item the symbol $\boldsymbol{1}$ is a term of type $\emptyset$,
                \item for $i \in [k]$, the symbol $\mathsf{v}_i$ is a term of type $\emptyset$,
                \item for terms $t_1, t_2$ of type $I \subseteq [k]$ and $a \in \mathbb{Q}$, the expression $a t_1 + t_2$ is a term of type $I$,
                \item for terms $t_1, t_2$ of type $\emptyset$, the expression  $t_1 \odot t_2$ is a term of type $\emptyset$,
                \item for a term $t$ of type $I$ and $i \in I$, the expression $\zeta_i t$ is a term of type $ I \setminus \{i\}$,
                \item for a term $t$ of type $I$ and $i \in [k] \setminus I$,
                the expression $\mu_i t$ is a term of type $I \cup \{i\}$,
                \item for a term $t$ of type $I$ and $i,j \in I$ with $i \neq j$, the expression $\beta_{ij} t$ is a term of type $I$,
                \item for a term $t$ of type $I$ and $i, j \in [k] \setminus I$ with $i \neq j$, the expression $\rho_{i \to j}t$ is a term of type $I$. 
        \end{enumerate}
\end{definition}

Next we define the semantics of $\mathsf{CW}^k$.
For a graph $G$, on which the $\mathsf{CW}^k$-terms are to be evaluated,
write $P(G)$ for the powerset of vertex set $V(G)$. 
Each term~$t$ of $\mathsf{CW}^k$ is associated with
a vector $\den{t}_G \in \mathbb{Q}^{P(G)^k}$.

\begin{definition}[Semantics of $\mathsf{CW}^k$]\label{def:semantics}
        Let $G$ be a graph.
        Let $t, t_1, t_2$ be $\mathsf{CW}^k$-terms such that the following operations are well-defined.
        Let $a \in \mathbb{Q}$.
        Let $\boldsymbol{X} = (X_1, \dots, X_k) \in P(G)^k$.
        \begin{enumerate}
                \item $\den{\boldsymbol{1}}_G$ is the all-ones vector,
                \item for $i \in [k]$, $\den{\mathsf{v}_i}_G$ is the vector $(X_1, \dots, X_k) \mapsto |X_i|$,
                \item $\den{at_1 + t_2}_G = a \den{t_1}_G + \den{t_2}_G$,
                \item $\den{t_1 \odot t_2}_G = \den{t_1}_G \odot \den{t_2}_G$, where $\den{t_1}_G \odot \den{t_2}_G$ is the pointwise product of the vectors $\den{t_1}_G, \den{t_2}_G$,
                \item $\den{\zeta_i t}_G(\boldsymbol{X}) = \sum_{Y \subseteq X_i} \den{t}_G(\boldsymbol{X}[i/Y])$,
                \item $\den{\mu_i t}_G(\boldsymbol{X}) = \sum_{Y \subseteq X_i} (-1)^{|X_i \setminus Y|}\den{t}_G(\boldsymbol{X}[i/Y])$,
                \item $\den{\beta_{ij} t}_G(\boldsymbol{X}) =
                \begin{cases}
                    \den{t}_G(\boldsymbol{X}), & \text{if } X_i \cap X_j = \emptyset \text{ and } G[X_i, X_j] \text{ is complete bipartite},\\
                    0, & \text{otherwise},
                \end{cases}$
                \item $\den{\rho_{i \to j} t}_G(\boldsymbol{X}) = \den{t}_G(\boldsymbol{X}[i/X_j])$.
        \end{enumerate}
\end{definition}

We define the operator $\eta_{ij}$ as syntactic sugar: If $t$ is a term of type~$I$ and $i \neq j$ are not in $I$, then
\[
    \eta_{ij} t \coloneqq \zeta_i \zeta_j \beta_{ij} \mu_i\mu_j t
\]
is a term of type $I$.
This definition mirrors \cref{eq:join}.
By induction on the cliquewidth expression, it follows from \cref{lem:identities} that there exists, for every $k$-partitioned graph $\boldsymbol{F} \in \mathfrak{A}_k$ of cliquewidth~$\leq k$, a $\mathsf{CW}^k$-term $t$ such that $\den{t}_G = \boldsymbol{F}_G$.
Here, the operations of a cliquewidth expression are replaced by the homonymous $\mathsf{CW}^k$-term operations.
By defining $\mathsf{CW}^k$-equivalence as indicated by \cref{obs:drop-color}, it follows that $\mathsf{CW}^k$-equivalent graphs are homomorphism indistinguishable over all graphs of cliquewidth~$\leq k$, i.e., one of the directions of \cref{thm:main-characterisation}.

\begin{definition}[Equivalence]
        Let $k \geq 1$.
        Two graphs $G$ and $H$ are \emph{$\mathsf{CW}^k$-equivalent} if,
        for every $\mathsf{CW}^k$-term $t$ of type $\emptyset$, it holds that
        \(
                \den{t}_G(V(G), \dots, V(G)) = \den{t}_H(V(H), \dots, V(H)).
        \)
\end{definition}

For the missing direction of \cref{thm:main-characterisation},
we characterize the vector space of $\mathsf{CW}^k$-term evaluations $\den{t}_G \in \mathbb{Q}^{P(G)^k}$.
To that end,
let $\boldsymbol{F} = (F; C_1, \dots, C_k)$ be a $k$-partitioned graph.
For a set $I \subseteq [k]$ and a graph $G$, we define the following vector in $\mathbb{Q}^{P(G)^k}$.
For $\boldsymbol{X} \in P(G)^k$,
\begin{equation}
    \boldsymbol{F}^I_G(\boldsymbol{X})
    = \left|\left\{h \colon F\to G \mid h(C_i) = X_i \text{ for } i \in I \text{ and } h(C_i) \subseteq X_i \text{ for } i \in [k] \setminus I\right\}\right|.
\end{equation}
That is, $\boldsymbol{F}^I_G(\boldsymbol{X})$ counts the homomorphisms from $F$ to $G$ that, for $i \in [k]$, map the set $C_i$ into or onto $X_i$, depending on whether $i \notin I$ or $i \in I$.
Note that $\boldsymbol{F}^\emptyset_G$ coincides with $\boldsymbol{F}_G$ as defined in \cref{def:hom-vector}.
\begin{lemma}\label{lem:spaces-coincide}
    For every $I \subseteq [k]$,
    the space spanned by the vectors $\den{t}_G$ for $\mathsf{CW}^k$-terms of type~$I$ coincides with the space spanned by the vectors $\boldsymbol{F}^I_G$ for $\boldsymbol{F} \in \mathfrak{A}_k$.
    Furthermore, the coefficients when writing $\den{t}_G$ as a linear combination of the $\boldsymbol{F}^I_G$, and vice versa, depend only on~$t$ and~$\boldsymbol{F}$ and not on~$G$.
\end{lemma}
\begin{proof}
    For the backward inclusion, we argue by induction on the cliquewidth expression using \cref{lem:identities}.
    Crucially, it holds that $\zeta_i \boldsymbol{F}^I_G = \boldsymbol{F}^{I \setminus \{i\}}_G$ if $i \in I$
    and $\mu_i \boldsymbol{F}^I_G = \boldsymbol{F}^{I \cup \{i\}}_G$ if $i \not\in I$.
    Here, we view $\zeta_i$ and $\mu_i$ as linear maps $\mathbb{Q}^{P(G)^k} \to \mathbb{Q}^{P(G)^k}$ as given in \cref{def:semantics}.
    For the converse inclusion, we argue by induction on \cref{def:syntax}.
    Here, it is crucial that all unary operations in \cref{def:semantics} are linear and can thus be applied to each term $\boldsymbol{F}^I_G$ in a linear combination $\den{t}_G = \sum \alpha_{\boldsymbol{F}} \boldsymbol{F}^I_G$ individually.
    The pointwise product $\odot$ is bilinear; it is applied only at type~$I = \emptyset$, where $\boldsymbol{F}^\emptyset_G \odot (\boldsymbol{F}')^\emptyset_G = (\boldsymbol{F} \oplus \boldsymbol{F}')^\emptyset_G$ by~\eqref{eq:disjoint-union}, so it can be applied to each pair of summands individually.
\end{proof}

\cref{lem:spaces-coincide} yields the remaining implications of \cref{thm:main-characterisation}. 
Next, we turn to the proof of \cref{thm:dense-to-sparse}.
The key observation is that in a $K_{t,t}$-subgraph-free graph~$G$, \cref{eq:join} can be simplified to involve only sums over vertex subsets of bounded size.
To that end, we first note that, if $A \subseteq V(G)$ is a vertex subset, then
\begin{equation}\label{eq:small}
    |A| < t \qquad \text{or} \qquad |N_G^\cap(A)| < t
\end{equation}
where $N_G^\cap(A)$ is the set of vertices in $G$ that are adjacent to every vertex in $A$.
Indeed, if $|A| \geq t$ and $|N_G^\cap(A)| \geq t$, then $A$ and $N_G^\cap(A)$ would yield a $K_{t,t}$-subgraph as $A \cap N_G^\cap(A) = \emptyset$ because $G$ is loopless.

\begin{lemma}\label{lem:ktt-free-join}
    For a $K_{t,t}$-subgraph-free $G$ and a $\mathsf{CW}^k$-term~$s$ of type~$I$ with $i \neq j$ not in $I$,
    \begin{align*}
        \den{\eta_{ij} s}_G(\boldsymbol{X})
        =& \sum_{\substack{A \subseteq X_i \\ |A| < t}} \den{\mu_i s}_G(\boldsymbol{X}[i/A, j/(X_j \cap N^\cap_G(A))])\\
        &+ \sum_{\substack{B \subseteq X_j \\ |B| < t}} \den{\mu_j s}_G(\boldsymbol{X}[i/(X_i \cap N^\cap_G(B)), j/B])\\
        &- \sum_{\substack{A \subseteq X_i, B \subseteq X_j \\ |A| < t,  |B| < t}} \den{\beta_{ij}\mu_i\mu_j s}_G(\boldsymbol{X}[i/A, j/B]).
    \end{align*}
\end{lemma}
\begin{proof}
Let $q \coloneqq \mu_i \mu_j s$.
By the definition of $\eta_{ij}$ (cf.~\eqref{eq:join}) and the semantics of $\zeta_i, \zeta_j, \beta_{ij}$,
$\den{\eta_{ij} s}_G(\boldsymbol{X})$ is the sum of $\den{q}_G(\boldsymbol{X}[i/A, j/B])$ over biclique pairs
$A\subseteq X_i$, $B\subseteq X_j$.
By \Cref{eq:small}, these pairs
are the union of those with $|A|< t$ and those with $|B| < t$.
The three summands in the lemma statement correspond to the pairs with $|A|< t$ and $|B|$ arbitrary, with $|A|$ arbitrary and $|B| < t$, and $|A|, |B| < t$. The signs are due to inclusion--exclusion.
In the first case, the biclique condition rewrites as $B \subseteq X_j \cap N_G^\cap(A)$.
\end{proof}

\cref{lem:ktt-free-join} is the key ingredient for proving \cref{thm:dense-to-sparse}.
Before we conduct the proof, we record a lemma, which is implicit in the work of \textcite{dvorak_2010_homomorphisms}.

\begin{lemma}[\cite{dvorak_2010_homomorphisms}]\label{lem:tw-compile}
    Let $G$ be a graph and $k \geq 0$.
    Every function $f \colon V(G)^{\leq k} \to \mathbb{Q}$ formed 
    \begin{enumerate}
        \item from the indicators $[\phi(x_1, \dots, x_k)] \colon V(G)^{\leq k} \to \{0,1\}$ for quantifier-free first-order formulas $\phi$ via
        \item linear combinations and pointwise multiplication, and
        \item constructing the arity-$(k-1)$ function $\boldsymbol{x} \mapsto \sum_{v \in V(G)} f(\boldsymbol{x}[i/v])$ from the arity-$k$ function~$f$ for $i \in [k]$
    \end{enumerate}
    can be written as a linear combination of homomorphism counting functions $\boldsymbol{x} \mapsto \hom(F, G; \boldsymbol{u} \mapsto \boldsymbol{x})$ where $F$ is a graph of treewidth~$< k$ and
    the vertices $\boldsymbol{u} \in V(F)^{\leq k}$ occur together in one bag of $F$'s tree decomposition.
\end{lemma}

As a first step, we derive a monadic version of \cref{lem:tw-compile} with quantification over small sets.

\begin{lemma}\label{lem:cw-compile}
    Let $G$ be a graph and $k,t \geq 0$.
    Every function $f \colon \binom{V(G)}{\leq t}^{\leq k} \to \mathbb{Q}$ formed 
    \begin{enumerate}
        \item from the indicators $[\phi(X_1, \dots, X_k)] \colon \binom{V(G)}{\leq t}^{\leq k} \to \{0,1\}$ for quantifier-free $\mathsf{MSO}$-formulas $\phi$ with predicates $X \subseteq Y$, $\mathsf{Biclique}(X, Y)$, $\mathsf{size}_\ell(X)$ for $\ell \in \mathbb{N}$; in addition to the variables $X_1, \dots, X_k$, the formula has access to $N_G^\cap(X_1), \dots, N_G^\cap(X_k)$.
        \item linear combinations and pointwise multiplication, and
        \item constructing the arity-$(k-1)$ function $\boldsymbol{X} \mapsto \sum_{A \subseteq V(G), |A| < t} f(\boldsymbol{X}[i/A])$ from the arity-$k$ function~$f$ for $i \in [k]$
    \end{enumerate}
    can be written as a function as in \cref{lem:tw-compile} for $k' \coloneqq kt$.
\end{lemma}
\begin{proof}
    Replace each monadic variable $X$ by $< t$ first-order variables.
    Membership and distinctness etc.\ can be encoded using quantifier-free first-order predicates.
\end{proof}

\begin{proof}[Proof of \cref{thm:dense-to-sparse}]
    Let $G$ and $H$ be $K_{t,t}$-subgraph-free graphs.
    By induction on the cliquewidth operations constructing a $k$-partitioned graph $\boldsymbol{F} \in \mathfrak{A}_k$,
    we show that $\boldsymbol{F}_G$ and $\boldsymbol{F}_H$ can be written as a function as in \cref{lem:cw-compile}
    with $4k$ monadic variables.
    Throughout the induction, we maintain the invariant that the constraint~$X_\ell$ on the image of each color~$\ell \in [k]$ is represented using small sets only:
    by~\cref{eq:small}, $X_\ell$ is either small itself or it is the common neighborhood of a small set.
    Consequently, the functions only have access to monadic variables for small sets.
    By using $N_G^\cap(\emptyset) = V(G)$ and \cref{obs:drop-color},
    we can read off the final homomorphism count even though $V(G)$ is larger than~$t$.
   
    In the base case $\boldsymbol{F} = \bullet_i$, we use the size predicates.
    Disjoint union is handled by multiplication~\eqref{eq:disjoint-union},
    renaming by variable substitution~\eqref{eq:rename}.
    It remains to handle biclique insertions via \cref{lem:ktt-free-join}.
To that end, observe that
    \[
        N_G^\cap(A) \cap N_G^\cap(B) = N_G^\cap(A \cup B)
    \]
    and by \cref{eq:small} at least one of $A \cup B$ or $N_G^\cap(A \cup B)$ is small.
    In the case
    \[
        A \cap N_G^\cap(B),
    \]
    the set $A$ is small by induction and hence so is this intersection.

    It remains to bound the number of small monadic variables kept in any scope simultaneously.
    We use four such variables per cliquewidth color.
    Two are needed for $A$ and $B$ in \cref{lem:ktt-free-join},
    two for the M\"obius inversion carried out by $\mu_i$ and $\mu_j$.
    All sets are small by \cref{lem:ktt-free-join}.

    By \cref{lem:cw-compile} with $4k$ monadic variables, i.e.\ for $k' = 4kt$, and \cref{lem:tw-compile}, the value $\hom(F, G) = \boldsymbol{F}_G(V(G), \dots, V(G))$ is a linear combination, with coefficients independent of~$G$, of homomorphism counts from graphs of treewidth $< 4kt$.
    The same linear combination computes $\hom(F, H)$.
    Since $G$ and $H$ are not distinguished by the $4kt$-dimensional Weisfeiler--Leman algorithm, they agree on all these homomorphism counts~\cite{dvorak_2010_homomorphisms}, hence $\hom(F, G) = \hom(F, H)$ for every $\boldsymbol{F} \in \mathfrak{A}_k$.
    Thus $G$ and $H$ are homomorphism indistinguishable over all graphs of cliquewidth~$\leq k$ and therefore not distinguished by the $k$-dimensional dense Weisfeiler--Leman algorithm by \cref{thm:dwl-hom}.
\end{proof} 
\section{Conclusion}

We have given the first characterization and algorithm for homomorphism indistinguishability over dense graph classes.
Our algorithm for deciding homomorphism indistinguishability over the class of graphs of cliquewidth~$\leq k$ represents a dense analogue of the Weisfeiler--Leman algorithm.
Various questions regarding dense homomorphism indistinguishability remain open:

The complexity upper bounds from \cref{thm:cw-main,thm:main-clique-width-mso,thm:main-linear-clique-width-mso} match what would be expected for succinctly represented version of the corresponding sparse problems \cite{grohe_equivalence_1999,cerny_homomorphism_2026}, see \cref{fig:complexity-overview}.
Our \cref{thm:pspace_hard} shows that the linear cliquewidth algorithm is best possible.
For \cref{thm:cw-main,thm:main-clique-width-mso}, we expect that reducing succinct Boolean circuit evaluation and polynomial identity testing problems can show optimality.
A first step in this direction was taken in \cref{thm:cographs}.
Furthermore, it would be interesting to optimize the base~$4$ of the exponent in \cref{thm:main-runtime} or show hardness under the exponential time hypothesis.
 
Besides complexity, it would be interesting to understand the distinguishing power of the dense Weisfeiler--Leman algorithm better.
How does our logic~$\mathsf{CW}^k$ compare to low-rank $\mathsf{MSO}$
\cite{bojanczyk_low_2026}? 
Is there a Spoiler--Duplicator game \cite{hella_logical_1996} for $\mathsf{CW}^k$?
In contrast to treewidth, cliquewidth is functionally equivalent to various equally well-motivated graph parameters.
Perhaps there is a more natural characterization of homomorphism indistinguishability over the graphs of rankwidth~$\leq k$, multi-cliquewidth~$\leq k$, or fusionwidth~$\leq k$ \cite{oum_approximating_2006,furer_multi-clique-width_2017,furer2014natural}.

Finally, our results are only the starting point for a theory of homomorphism indistinguishability over dense graph classes:
Which properties of $\mathcal{F}$ ensure that $\equiv_{\mathcal{F}}$ is decidable, see the corresponding conjecture \cite{seppelt2024algorithmic} for minor-closed~$\mathcal{F}$?
Here, \cref{thm:main-clique-width-mso} subsumes all known sufficient conditions.
Moreover, is \cref{thm:dense-to-sparse} an instance of a more general dense-to-sparse collapse phenomenon? 
That is, for which graph classes $\mathcal{F}$, does there exist a function $f \colon \mathbb{N} \to \mathbb{N}$
such that $K_{t,t}$-subgraph-free graphs $G$ and $H$ are homomorphism indistinguishable over $\mathcal{F}$ if, and only if, they are homomorphism indistinguishable over the $K_{f(t),f(t)}$-subgraph-free graphs in $\mathcal{F}$?
Our \cref{thm:dense-to-sparse} proves that such a function exists for the classes $\mathcal{F}_k$ of graphs of cliquewidth~$\leq k$.
More precisely,
is the class of graphs of cliquewidth~$\leq k$ homomorphism distinguishing closed \cite{roberson_oddomorphisms_2022,neuen_distinguishing_2026}?

\section*{AI Statement}
\begin{itemize}
    \item After writing an initial draft of the content of \cref{sec:logic_first10}, this draft was given to ChatGPT 5.5 to streamline and simplify the content.
    \item ChatGPT 5.5 formulated a first version of proof of the $\mathsf{PSPACE}$ lower bound \cref{thm:pspace_hard} based on the reduction source and intuition, which we then modified and edited.
    \item ChatGPT 5.5 found an initial version of \cref{thm:cographs}.
    \item Claude Opus 4.8 wrote the code to generate \cref{fig:execution} (illustrating an execution of the dense WL algorithm) based
    on the algorithm description in \cref{sec:dense-weisfeiler--leman-algorithm}.
\end{itemize}

\section*{Acknowledgements}

We are grateful for early-stage discussions with Miko{\l}aj Boja\'nczyk and Szymon Toru\'nczyk.

Amir Nikabadi is supported by the Independent Research Fund Denmark (DFF), grant agreement number 2098-00012B.

Radu Curticapean, Tim Seppelt, and Ben Young are supported by the European Union (CountHom, 101077083). Views and opinions expressed are however those of the author(s) only and do not necessarily reflect those of the European Union or the European Research Council Executive Agency. Neither the European Union nor the granting authority can be held responsible for them.

\printbibliography

\begin{figure}
    \begin{center}
    \begin{subfigure}{.48\linewidth}
        \centering
        \includegraphics[width=.87\linewidth]{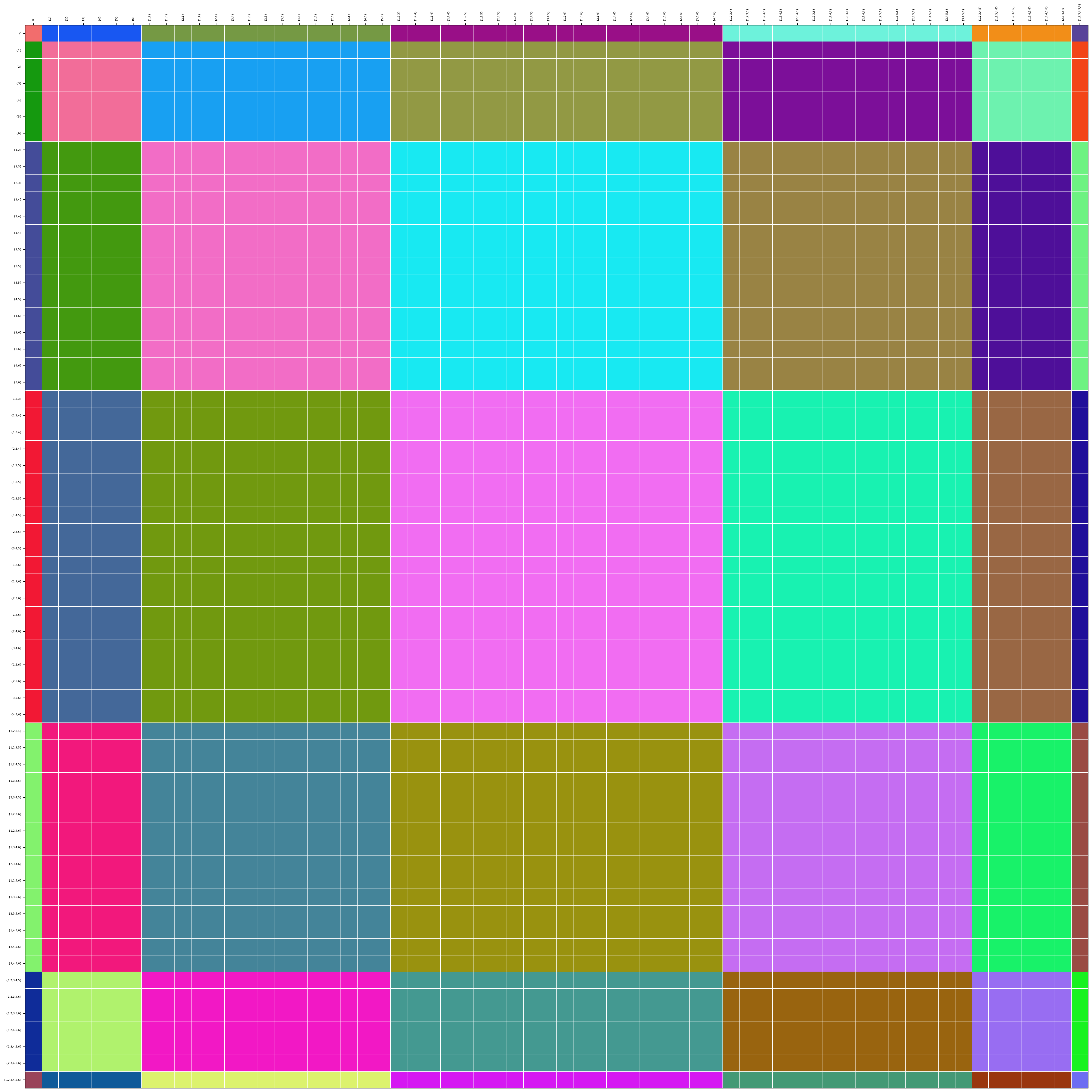}
        \caption{$C_3+C_3$, iteration~0, 49 classes}
    \end{subfigure}
    \begin{subfigure}{.48\linewidth}
        \centering
        \includegraphics[width=.87\linewidth]{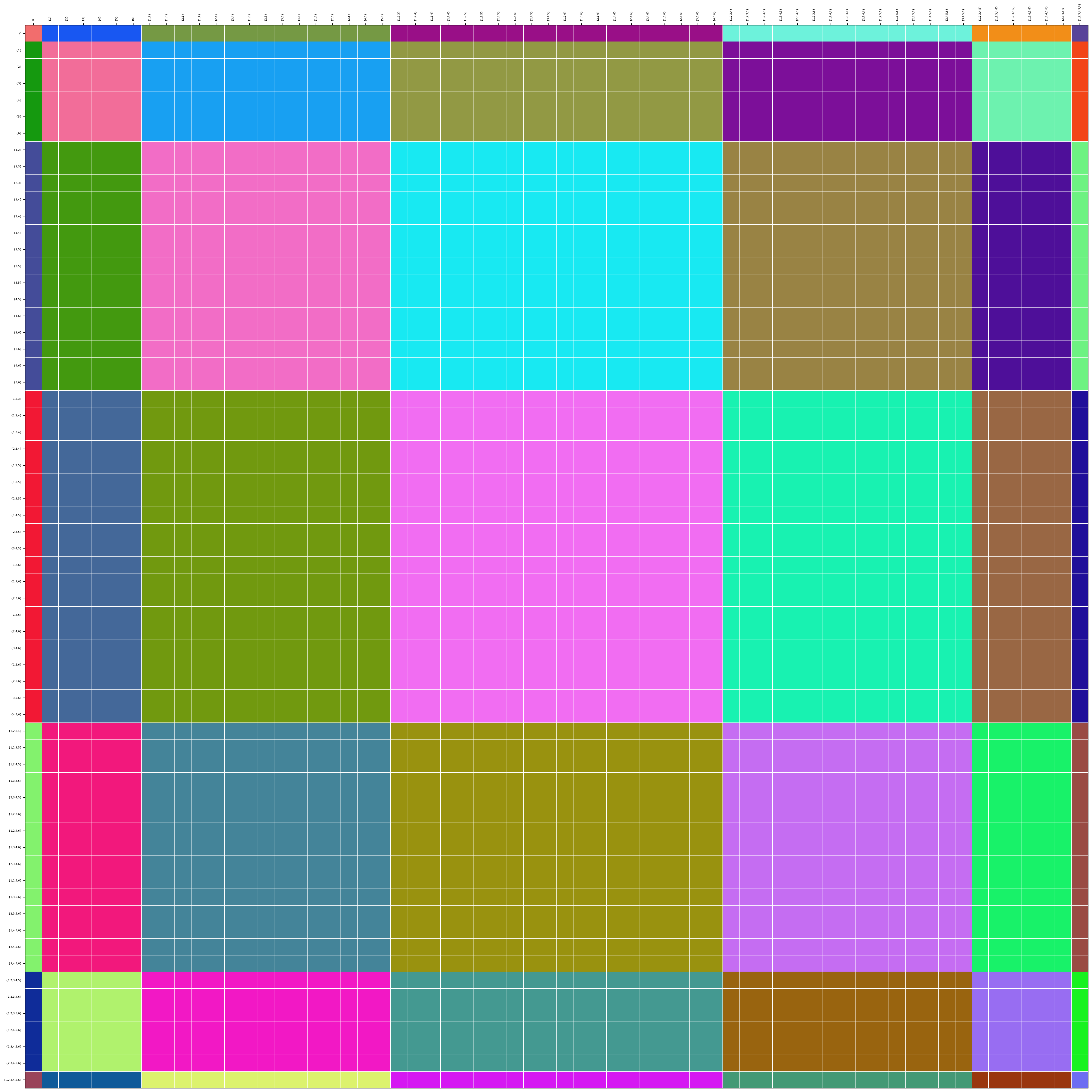}
        \caption{$C_6$, iteration~0, 49 classes}
    \end{subfigure}
    \begin{subfigure}{.48\linewidth}
        \centering
        \includegraphics[width=.87\linewidth]{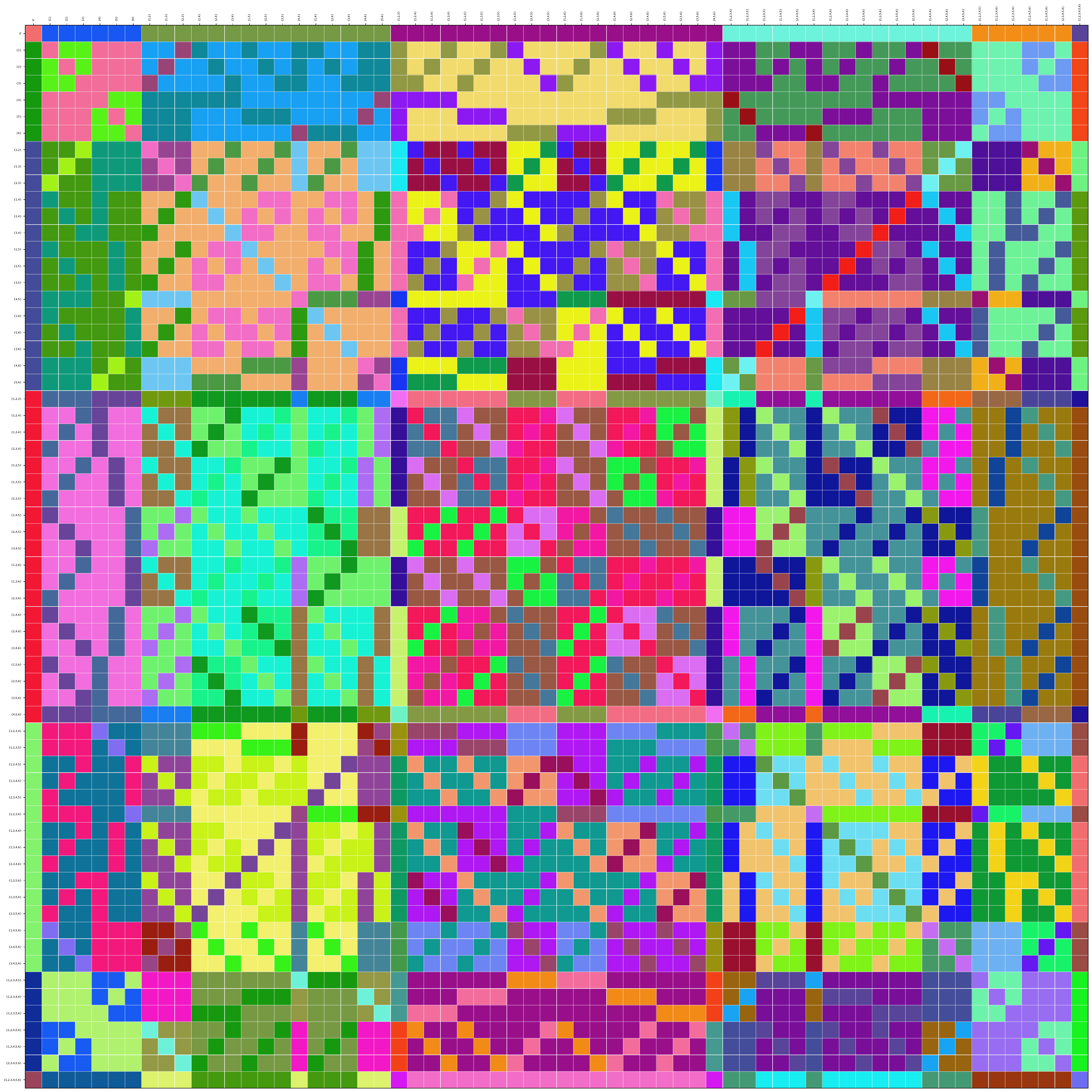}
        \caption{$C_3+C_3$, iteration~1, 162 classes}
    \end{subfigure}
    \begin{subfigure}{.48\linewidth}
        \centering
        \includegraphics[width=.87\linewidth]{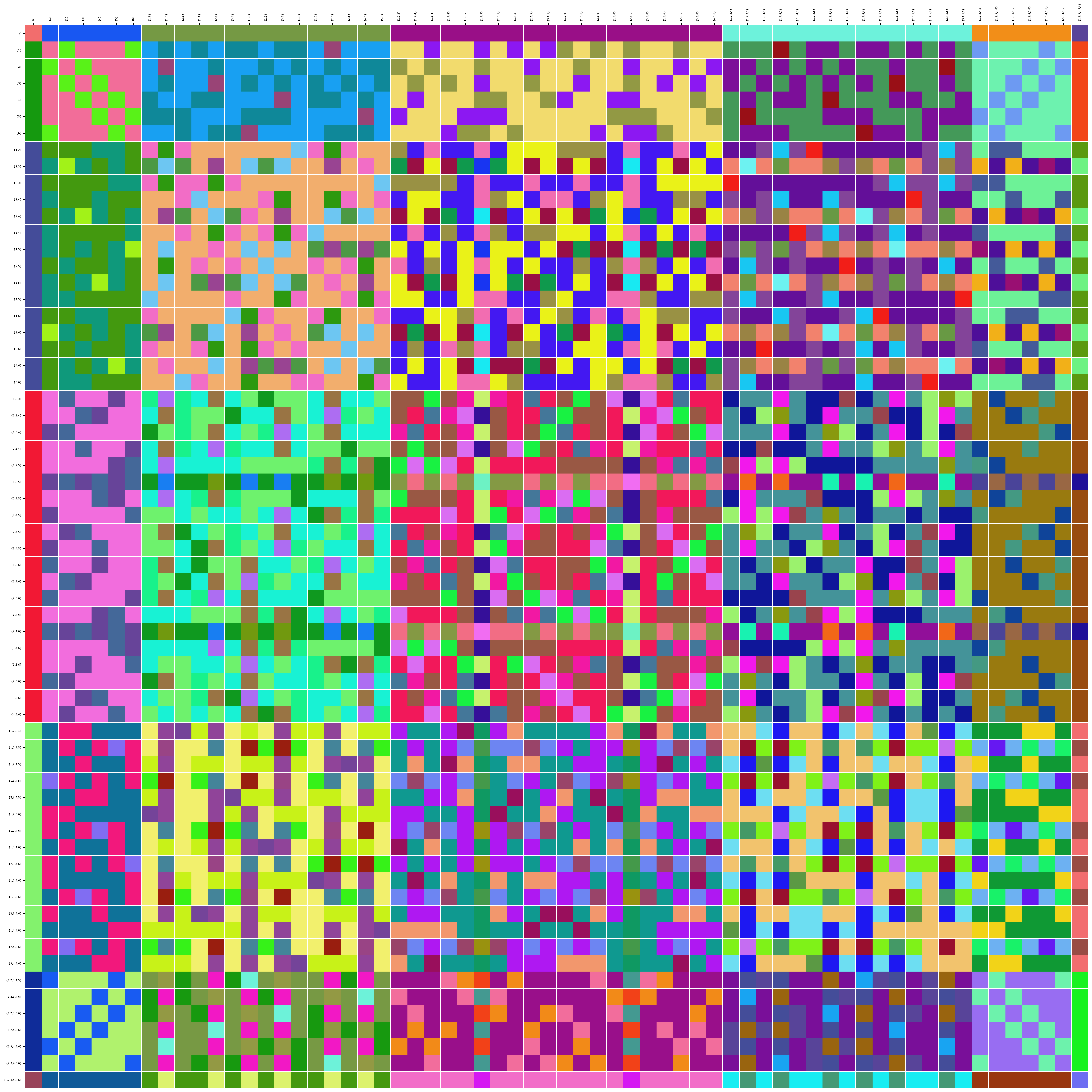}
        \caption{$C_6$, iteration~1, 162 classes}
    \end{subfigure}
    \begin{subfigure}{.48\linewidth}
        \centering
        \includegraphics[width=.87\linewidth]{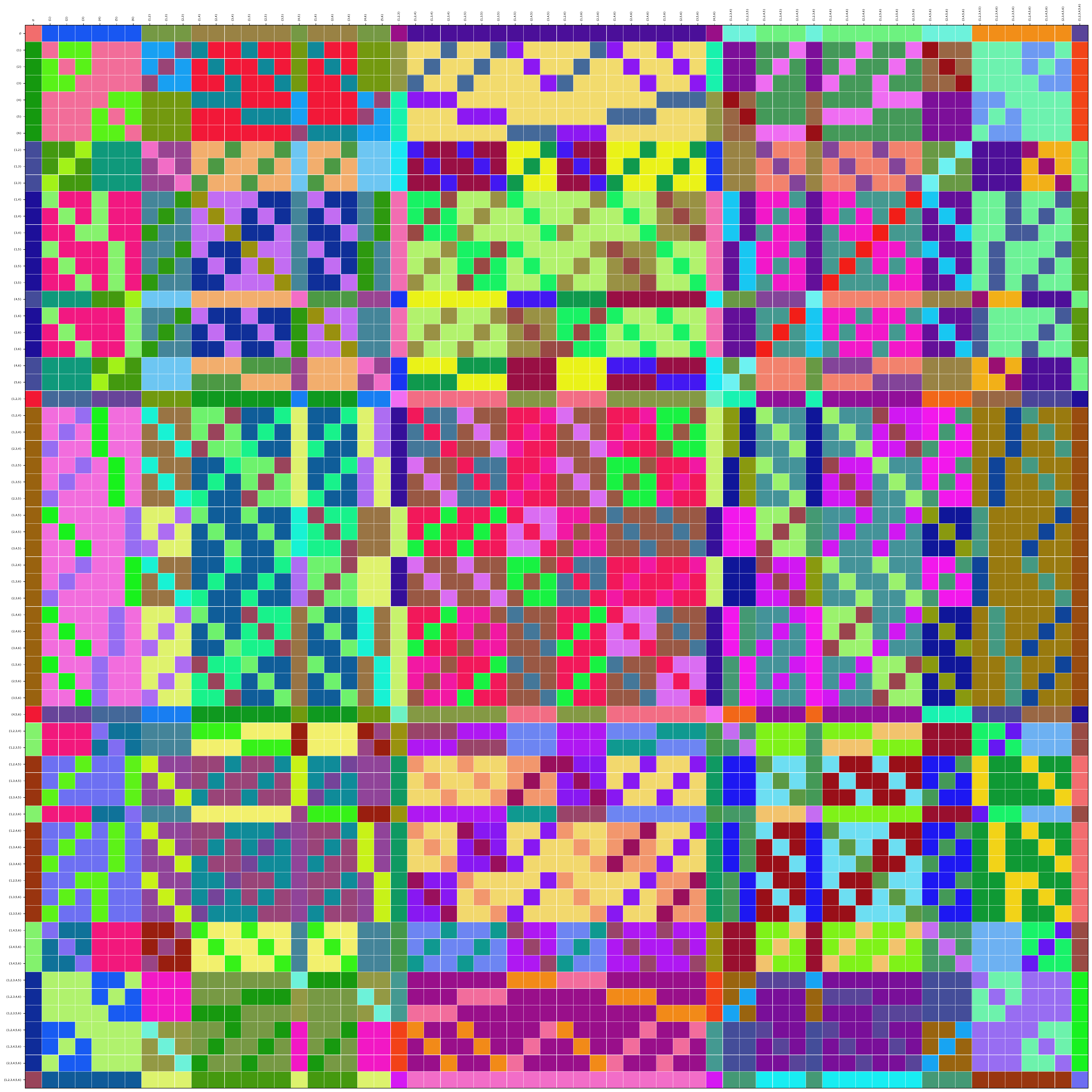}
        \caption{$C_3+C_3$, iteration~2, 200 classes}
    \end{subfigure}
    \begin{subfigure}{.48\linewidth}
        \centering
        \includegraphics[width=.87\linewidth]{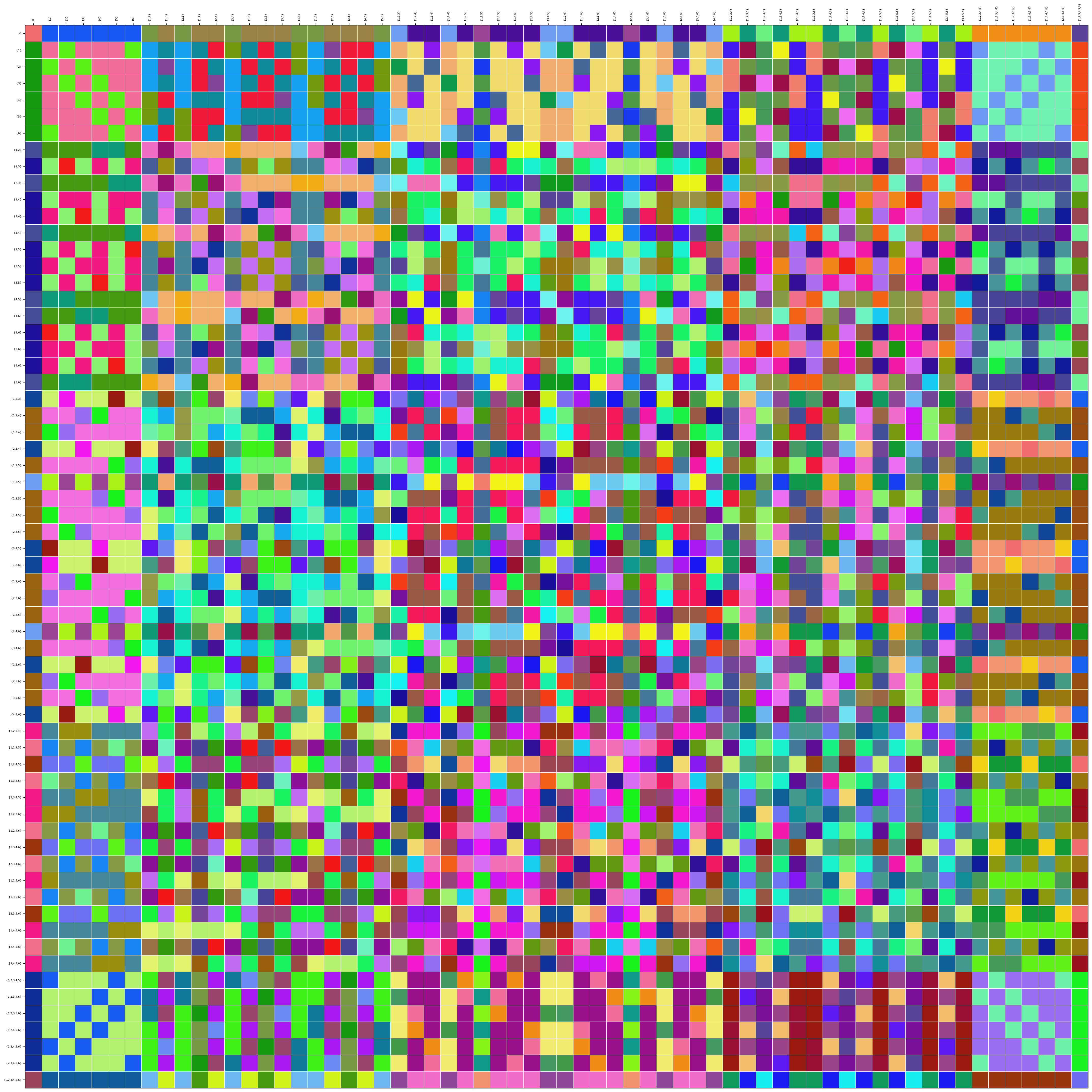}
        \caption{$C_6$, iteration~2, 329 classes}
    \end{subfigure}
    \end{center}
\end{figure}
\begin{figure}\ContinuedFloat
    \begin{center}
    \begin{subfigure}{.48\linewidth}
        \centering
        \includegraphics[width=.87\linewidth]{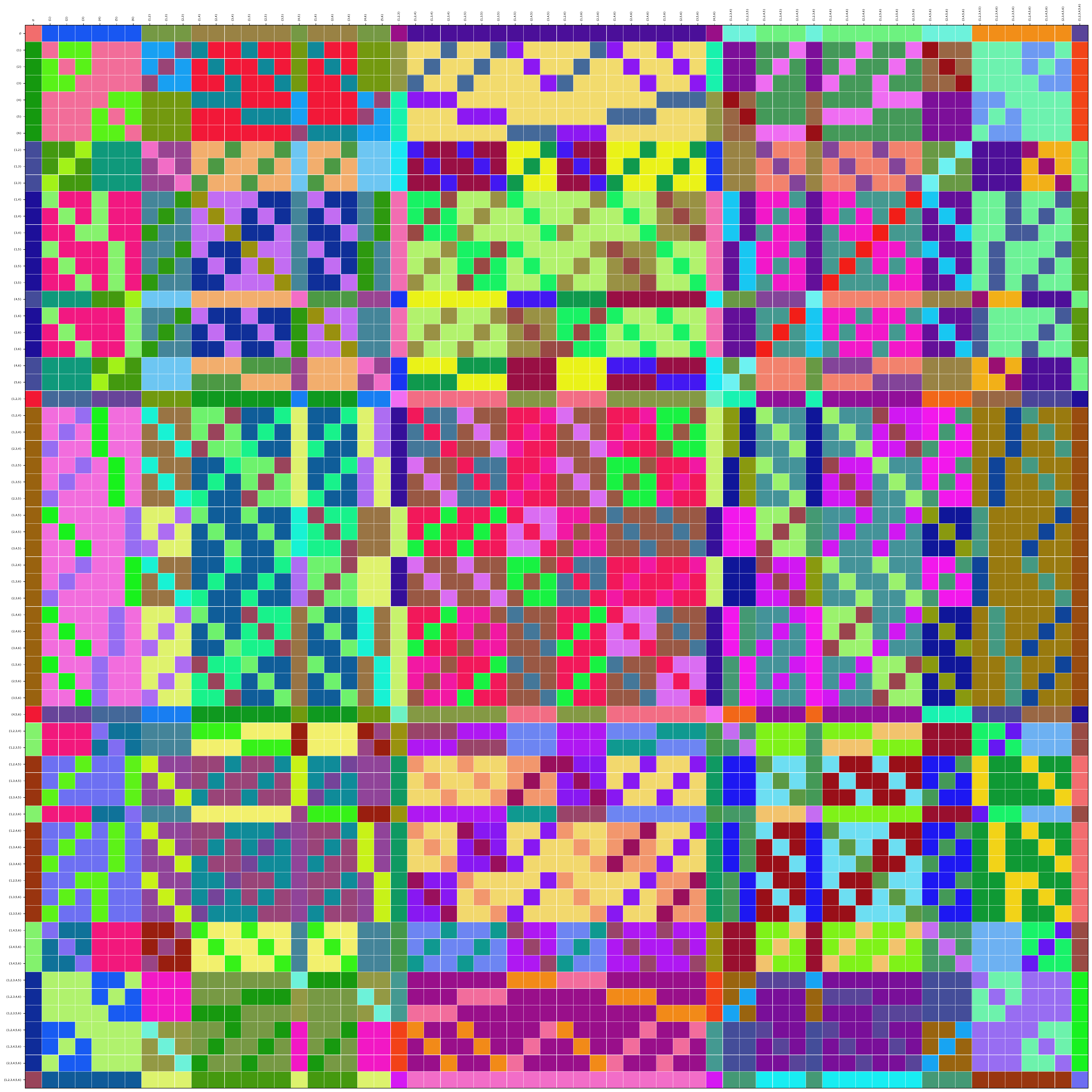}
        \caption{$C_3+C_3$, iteration~3, 200 classes}
        \label{fig:bottom_right_pixel1}
    \end{subfigure}
    \begin{subfigure}{.48\linewidth}
        \centering
        \includegraphics[width=.87\linewidth]{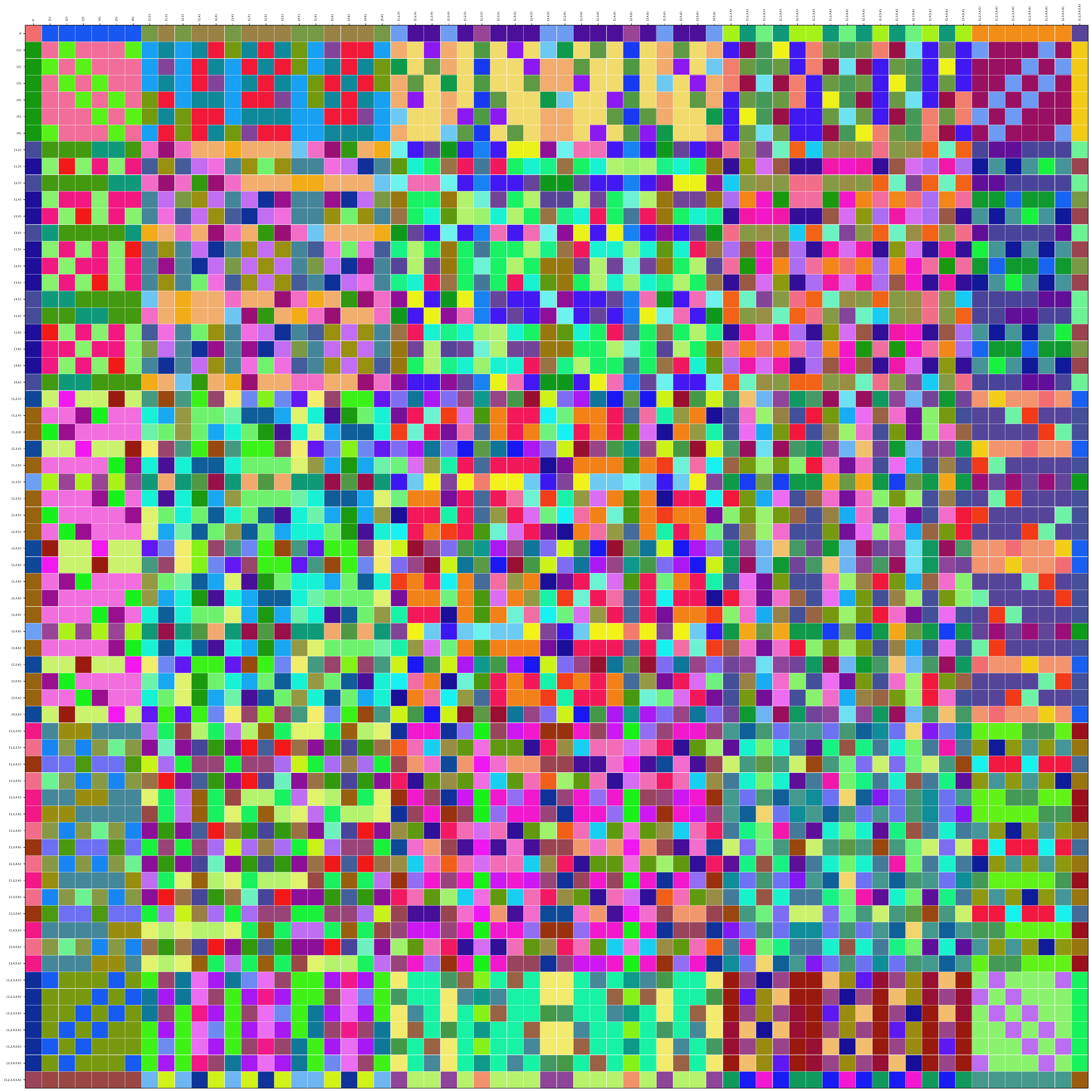}
        \caption{$C_6$, iteration~3, 329 classes}
        \label{fig:bottom_right_pixel2}
    \end{subfigure}
    \begin{subfigure}{.48\linewidth}
        \centering
        \includegraphics[width=.87\linewidth]{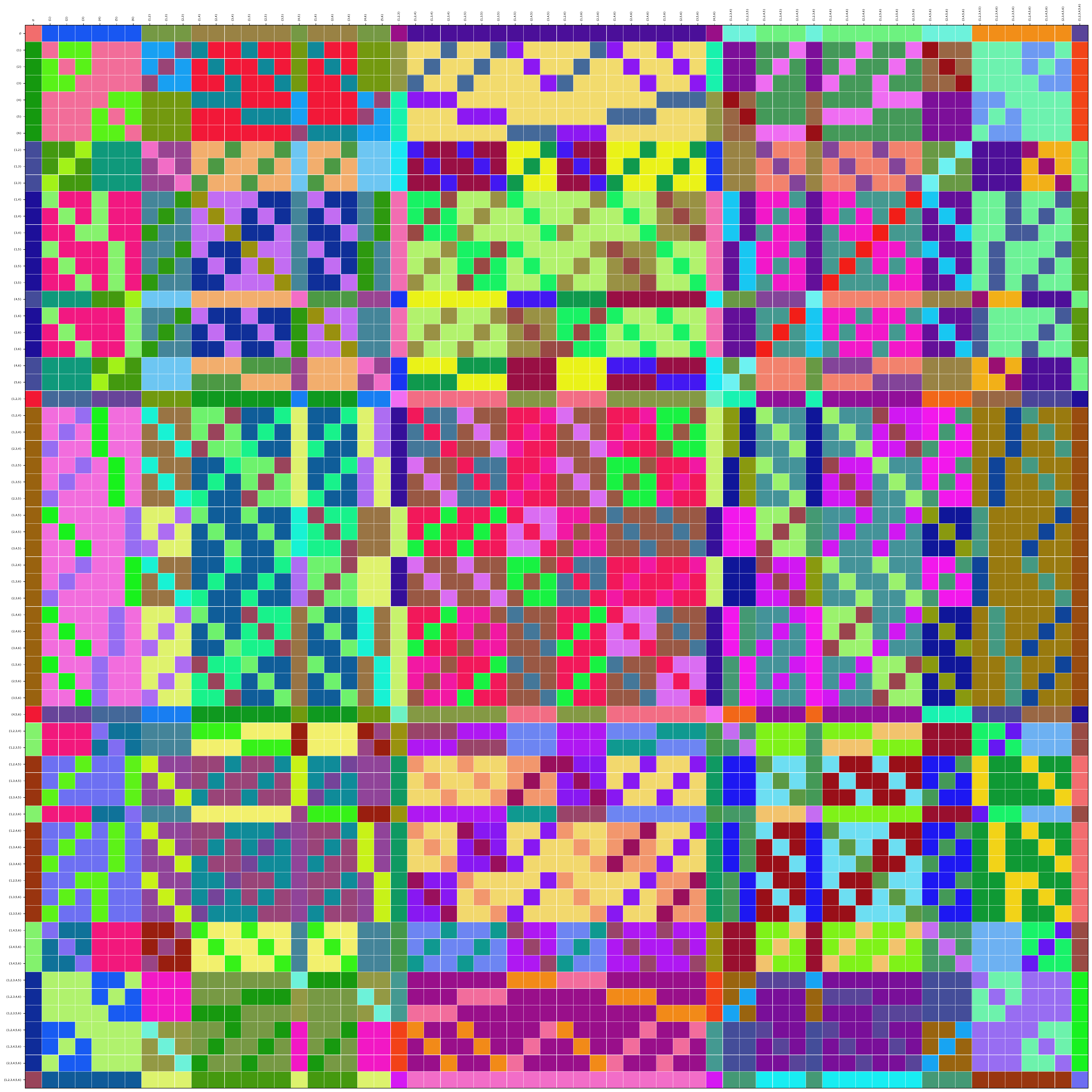}
        \caption{$C_3+C_3$, iteration~4, 200 classes}
    \end{subfigure}
    \begin{subfigure}{.48\linewidth}
        \centering
        \includegraphics[width=.87\linewidth]{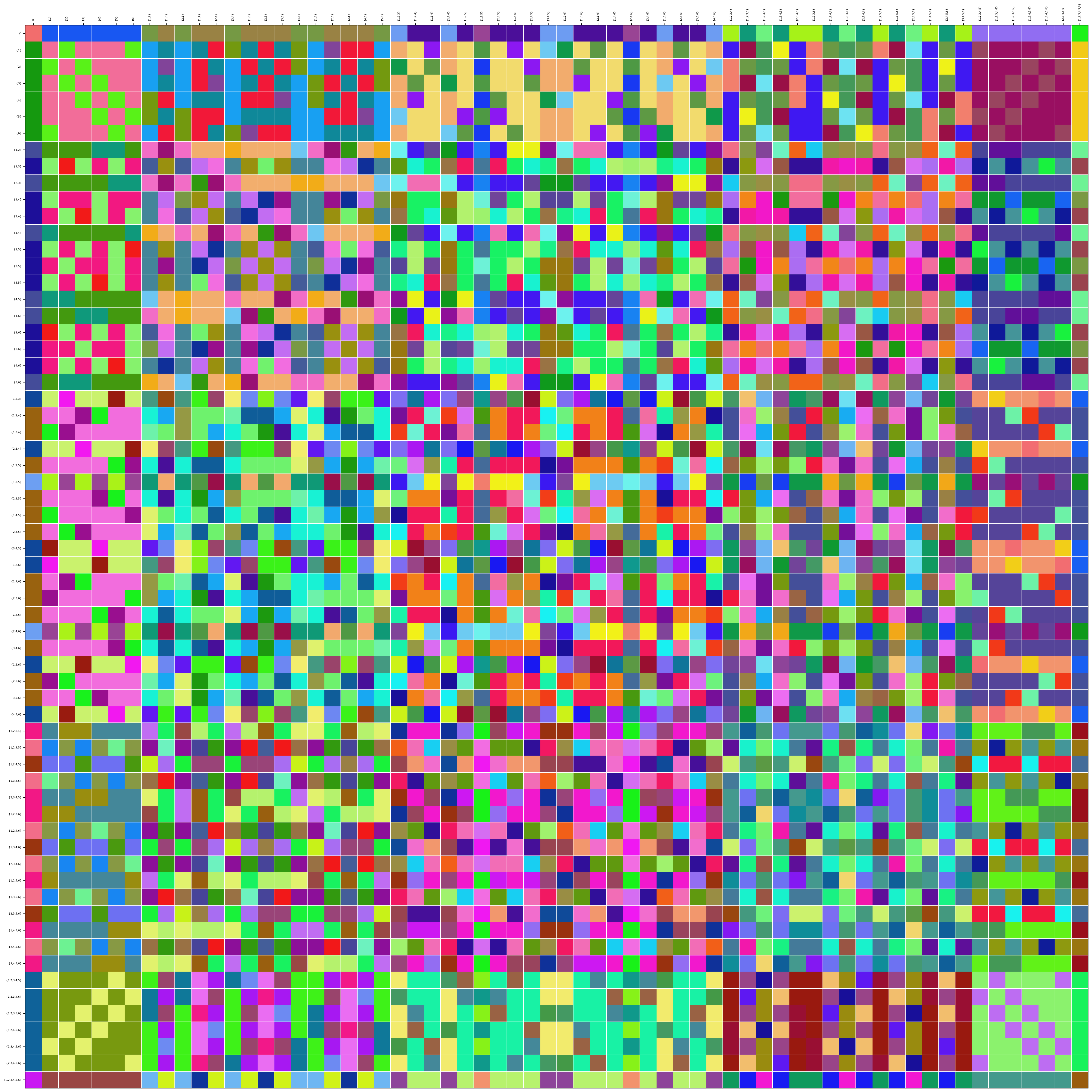}
        \caption{$C_6$, iteration~4, 329 classes}
    \end{subfigure}
    \end{center}
    \caption{A run of the 2-dimensional dense Weisfeiler--Leman algorithm on the graphs $C_3+C_3$ (left) and $C_6$ (right), which the 1-dimensional Weisfeiler--Leman algorithm fails to distinguish. Each
    array is $2^6 \times 2^6$, indexed by pairs of subsets of vertices of the respective graphs. The graphs are distinguished
    on iteration 3 (\cref{fig:bottom_right_pixel1,fig:bottom_right_pixel2}), when the bottom right pixels, indexed by the full vertex sets on both axes, become different colors. In iteration~4, the coloring stabilizes.}
    \label{fig:execution}
\end{figure}

\end{document}